\documentclass[journal,twocolumn]{IEEEtran}
\usepackage{amsfonts}    
\usepackage{color}
\usepackage{graphicx}
\usepackage[dvips]{epsfig}
\usepackage{graphics}
\usepackage{cite}
\usepackage{arydshln}
\usepackage{amsthm}
\usepackage{amsmath} 
\usepackage{amssymb}  
\usepackage[boxed,ruled,lined]{algorithm2e}
\usepackage{color}

\def\twon #1{\|#1\|}
\def\rainfty{\rightarrow\infty}
\def\ra{\rightarrow}

\def\argmin{\text{argmin}}

\def\bE{\mathbb{E}}

\def\bN{\mathbb{N}}

\def\bP{\mathbb{P}}

\def\bR{\mathbb{R}}

\def\cE{\mathcal{E}}
\def\cF{\mathcal{F}}
\def\cG{\mathcal{G}}
\def\cH{\mathcal{H}}

\def\cL{\mathcal{L}}

\def\cR{\mathcal{N}}

\def\cP{\mathcal{P}}

\def\cR{\mathcal{R}}
\def\cS{\mathcal{S}}

\def\cV{\mathcal{V}}

\def\cY{\mathcal{Y}}

\def \qed {\hfill \vrule height6pt width 6pt depth 0pt}
\def\bee{\begin{equation}}
\def\ene{\end{equation}}
\def\beq{\begin{eqnarray}}
\def\enq{\end{eqnarray}}

\newtheorem{assum}{Assumption}
\newtheorem{lem}{Lemma}
\newtheorem{rem}{Remark}
\newtheorem{cor}{Corollary}
\newtheorem{thm}{Theorem}
\newtheorem{exmp}{Example}

\newtheorem{defi}{Definition}

\def\bone{{\bf 1}}

\begin{document}
\title{Distributed Algorithms for Computation of Centrality Measures in Complex Networks}
\author{Keyou You, Roberto Tempo, and Li Qiu \thanks{This work was supported by the National Natural Science Foundation of China (61304038),  Tsinghua University Initiative Scientific Research Program, CNR International Joint Lab COOPS and Hong Kong Research Grants Council (618511).} \thanks{Keyou You is with the Department of Automation and TNList, Tsinghua University, 100084, China (email: youky@tsinghua.edu.cn).} \thanks{Roberto Tempo is with CNR-IEIIT, Politecnico di Torino, Torino, 10129, Italy (email: roberto.tempo@polito.it).} \thanks{Li Qiu is the Department of Electronic and Computer Engineering, The Hong Kong University of Science and Technology, Hong Kong, China (email: eeqiu@ust.hk).}}

\maketitle
\begin{abstract}
This paper is concerned with distributed computation of several commonly used centrality
measures in complex networks. In particular, we propose deterministic algorithms,
which converge in finite time, for the distributed computation of the degree, closeness and betweenness centrality measures in directed graphs. Regarding
eigenvector centrality, we consider the PageRank problem as its
typical variant, and design distributed randomized algorithms to compute
PageRank for both fixed and time-varying graphs. A key feature of the proposed
algorithms is that they do not require to know the network size, which can be simultaneously estimated at every node, and that they are clock-free.
To address the PageRank problem of time-varying graphs, we introduce the novel concept of persistent graph, which eliminates the effect of spamming nodes. Moreover, we prove that these algorithms converge almost surely and in the sense of $L^p$. Finally, the effectiveness of the proposed algorithms is illustrated via extensive simulations using a classical benchmark.
\end{abstract}
\markboth{}%
{Shell \MakeLowercase{\textit{et al.}}: Bare Demo of IEEEtran.cls
for Journals}
\begin{IEEEkeywords}  Complex networks, centrality measures, distributed computation, randomized algorithms, convergence properties.\end{IEEEkeywords}

\section{Introduction}

Centrality measures refer to indicators which identify the importance of nodes in a complex network. As first developed in social networks,  many of them were introduced to reflect their sociological origin \cite{newman2010networks}. Nowadays, they have become an important tool in network analysis, and are widely used for ranking the personal influence in a social network, the webpage popularity in the Internet, the fast spread of epidemic diseases, and the key infrastructure in urban networks. In fact, the ranking of a large number of objects is one of the most topical problems in information systems. Depending on the specific application, different centrality measures may be of interest. In this work, we are interested in the commonly used degree \cite{wasserman1994social}, closeness \cite{bavelas1950communication}, betweenness \cite{freeman1977set} and eigenvector \cite{bonacich1972factoring} centralities in complex networks.

As the network size becomes increasingly large, it is usually very difficult to compute centrality measures, except for the degree centrality which is of limited use. To address this issue, it is of great importance to design {\it distributed algorithms} with good scalability properties for their computation, where each node evaluates centralities by only using {\it local} interactions. Although distributed algorithms may play a significant role in alleviating the computational burden,  the access to limited information renders it challenging to ensure that each node provides its exact centrality.
\textcolor[rgb]{0,0,1}{This requires a rigorous and challenging analysis regarding convergence properties of these algorithms, in particular for the PageRank computation.}

\textcolor[rgb]{0,0,1}{We recall that other communities, such as
sociology, biology, physics, applied mathematics and computer science, see \cite{Ercsey,segarra,bonacich1987power,newman2010networks,easley2010networks,boldi2014axioms,friedkin1991theoretical,borgatti2005centrality,langville2011google}
and references therein, focused their attention on network centrality. In particular, within computer science, several authors studied the topic of distributed computing from the networking viewpoint \cite{SI}, which however provides a different viewpoint than that presented in our paper.}
We also remark that the computation of the degree, closeness and  betweenness centralities are closely related. For instance, calculating the betweenness and closeness centralities in a network requires calculating the shortest paths between all pairs of vertices. While the degree centrality is trivial, the computations of betweenness and eigenvector centralities have been extensively studied.  Numerous algorithms have been designed to compute the betweenness centrality, including Floyd-Warshall algorithm \cite{weisstein2008floyd}, Johnson's algorithm \cite{cormen2001introduction} and Brandes' algorithm \cite{brandes2001faster}.  On unweighted graphs with loops and multiple edges, calculating the betweenness centrality takes $O(|\cV|\cdot |\cE|)$ time using the classical Brandes' algorithm, where $|\cV|$ and $|\cE|$  denote the number of nodes and edges of a graph, respectively.  However, these  algorithms are {\em centralized} and rely on global information of the network.

Recently, distributed algorithms for computing the betweenness and closeness centralities in an undirected {\em tree} have been proposed in \cite{wang2013distributed,wang2014distributed} via a dynamical systems approach \textcolor[rgb]{0,0,1}{and in \cite{wang2015distributed}, where a scalable
algorithm for the computation of the closeness, based only on local interactions, is proposed.} In particular, every node computes its own centrality under only local interactions with its neighbors. In this paper, we propose finite-time convergent algorithms to distributedly compute the closeness centrality of a directed graph and the betweenness centrality of an oriented tree from the perspective of partitioning the network into multi-levels of neighbors. Our algorithms take advantage of the fact that a tree does not contain any loop, and therefore every pair of nodes has at most one shortest path.

Another important measure of centrality is eigenvector centrality, which is defined as the principal eigenvector of an adjacency matrix of the graph. It measures the influence of a node by exploiting the idea that connections to high-scoring nodes are more influential. That is, ``an important node is connected to important neighbors." In various applications, the notion of eigenvector centrality has been modified for networks that are not strongly connected, for example in systems biology
\cite{Vidyasagar:11}, Eigenfactor computation in bibliometrics \cite{Bergstrom:07}
and Web ranking \cite{brin1998anatomy}. In this case, a key idea is to introduce a so-called ``teleportation factor'' which includes a free parameter generally set to $0.15$, see \cite{ishii2014PageRank,ishii2012pagerank} for further details.
The resulting modified network then becomes strongly connected.  In particular, for Web ranking, this modification leads to the well-known PageRank \cite{langville2011google}, which has attracted significant interest  from the systems and control community \cite{ishii2014PageRank,polyak2012regularization,fercoq2013ergodic}. Our interest regarding the eigenvector centrality is particularly focused on the
distributed computation of PageRank.

Currently, the number of webpages in the Internet is incredibly large, and it is not even exactly known. This observation raises two interesting questions: (1) How to distributedly estimate the network size? Under local interactions, this problem is nontrivial as we cannot ensure to count every page, and some page(s) might be counted more than once.  (2) Without knowing the network size, how to compute the PageRank by only using  local interactions? To the best of our knowledge, these problems remain widely open. If the network is time-invariant, its size is known and a global clock is available,
several distributed randomized algorithms have been proposed in \cite{ishii2010distributed} for calculating the PageRank.  The idea lies in the design of the so-called distributed link matrices by exploiting sparsity of the hyperlink structure.  Then, each page is randomly selected to update its importance value  by interacting with those connected by hyperlinks. However, the randomization is based on an independent and identically distributed process (i.i.d.) and should be known to every page.
The need for this global information is indeed critical and the i.i.d. assumption is certainly not mild.
Furthermore, the algorithms are based on a time-averaging operation, which inevitably slows down convergence.
Other distributed PageRank algorithms are provided in  \cite{fercoq2013ergodic,ishii2014PageRank,ravazzi2015ergodic,nesterov} and references therein.

To suitably address the above questions, in this paper, we first reformulate the PageRank problem from the  least squares (LS) point of view, and then propose a randomized algorithm to incrementally compute the PageRank. Specifically, we consider a {\it Web random surfer} exploring the Internet. When browsing a webpage, the surfer incrementally updates an estimate of the PageRank by using importance values of the pages that have outgoing hyperlinks to the current page. Then, he/she will either randomly select an outgoing hyperlink of the current page and move to the page pointed by this link, or jump to an arbitrary page of the Internet, after which the PageRank estimate is updated again. This process results in a new type of the celebrated Kaczmarz algorithm \cite{zouzias2013randomized} with Markovian randomization, instead of a simpler i.i.d. randomization which may not adequately represent
a Web surfer model. In this case, the proof of convergence (both almost surely and in the $L^p$ sense) requires to exploit deep properties of Markov processes in the theory of stochastically time-varying systems \cite{guo1994stability}.

Remarkably, the proposed randomized algorithms can be conveniently implemented in a fully distributed manner, and each node simply maintains an estimate of the importance values of its neighbors and itself. In addition, every node only requires to know the number of neighboring nodes, rather than the total number of nodes in the network.  Interestingly, the network size is simultaneously estimated with probability one by an individual node. We point out that, even if the PageRank computation is reformulated as a least squares problem, existing distributed optimization algorithms \cite{necoara2013random,nedic2015distributed,iutzeler2015explicit,duchi2012dual,wei2013distributed,nesterov} are not directly applicable because they require the knowledge of the network size.

For the case of temporal networks described by time-varying graphs, where hyperlinks vary over time, but the network size is assumed to be constant,  an interesting problem is how to define the PageRank to measure the importance of each node. To the best of our knowledge, this problem has never been studied in the systems and control community. Therefore, a {\em persistent} graph is introduced to eliminate the effect of transient hyperlinks. Intuitively, PageRank should not be affected by spamming links, and the importance value of a spamming page should be negligible. Our approach is indeed very useful to deal with spamming nodes which create a large volume of hyperlinks in a short period.  A persistent graph adds large weights on persistent hyperlinks, and a larger weight on more recent hyperlinks. Then, the proposed incremental algorithm is generalized to address temporal networks with time-varying links and its convergence properties are also rigorously established.

In summary, the main contributions of this work are at least threefold. First, distributed algorithms with finite-time convergence are derived to compute the closeness and betweenness centralities in a unified setting. Second, we reformulate the PageRank problem as a LS problem and provide a new type of Markovian Kaczmarz algorithm, with rigorous convergence, to compute the PageRank. The algorithm can also be distributedly implemented even with unknown network size. Third, a novel concept of persistent graph is adopted to effectively study the PageRank problem over time-varying networks.

The rest of the paper is organized as follows. In Section \ref{sec_prp},  we provide an overview of centrality measures and the PageRank problem in both fixed and time-varying graphs. In Section \ref{sec_dcdcbc}, we design deterministic algorithms to incrementally compute the degree, closeness and betweenness centrality measures. In Section \ref{sec_dcp}, the PageRank problem is  reformulated as a least squares problem, based on which incremental algorithms are introduced to distributedly compute the PageRank. In Section \ref{sec_rce}, the incremental algorithms are randomized by mimicking the behavior of a random surfer. We also prove convergence of the randomized incremental algorithms to the PageRank. The case of temporal networks with time-varying links and related convergence properties are studied in \ref{sec_tvg}. Simulation results for a classical benchmark are included in \ref{sec_simulation}. Some concluding remarks are drawn in Section \ref{sec_conclusion}.

{\bf Notation:}~ For any vector $x\in\bR^n$ and $p>0$, let $\twon{x}_p=\left(\sum_{i=1}^n|x_i|^p\right)^{1/p}$. If $p=2$, we simply write $\twon{x}=\twon{x}_2$, which denotes the Euclidean norm of a vector. If $A$ is a matrix, we use $\twon{A}_p$ to denote the matrix norm induced by the vector norm $\twon{\cdot}_p$, i.e., $  \twon{A}_p=\sup_{\twon{x}_p=1}\twon{Ax}_p$.
The norm $\twon{A}_{L^p}$ is defined by
$\twon{A}_{L^p}=(\bE\twon{A}_p^p)^{1/p}$, where $\bE$ denotes expectation. The symbol ${\bf 1}$ represents a vector with all elements equal to one. For a symmetric matrix $A$, the notation $A\geq 0~(A>0)$
means that $A$ is positive semi-definite (definite), and
the relation $A\geq B$ means that
$A-B\geq 0$.

 \section{Centrality Measures and PageRank}
\label{sec_prp}
Ranking nodes is a crucial question in network science and has attracted a lot of attention in several scientific communities. A key problem is how to discern the importance of each node, which can be formalized by defining a centrality measure. In this section, we review some basics of complex networks and their centrality measures before concentrating on their distributed computation by means of deterministic and randomized algorithms.

Given a graph $\cG=\{\cV,\cE\}$, where $\cV:=\{1,\ldots,N\}$ is the set of vertices (nodes) and $\cE$ is the set of edges, let  $A=(a_{ij})_{i,j=1}^N$ denote its adjacency matrix, i.e. $a_{ij} = 1$ if vertex $i$ is linked to vertex $j$ and $0$, otherwise.  That is, an ordered pair $(i,j)\in\cE$ if and only if $a_{ij}=1$. A forward {\em path} in $\cG$ is a sequence of nodes $(n_1,\ldots,n_k)$ with $k\ge 2$ and a corresponding sequence of edges $(n_i,n_{i+1})\in\cE$. The {\em distance} of a path is defined as the sum of $a_{ij}$ on the path. $\cG$ is {\em strongly connected} if any pair of vertices are reachable from each other via a forward path.  Let  $\cR_i^k$ be the set of nodes that are reachable from node $i$ via exactly $k$ directed edges, and let $\cL_i^k$ be  the set of nodes that are linked to node $i$ via exactly $k$ directed edges. Moreover, self-loops are not allowed, i.e., $a_{ii}=0$ for all $i\in\cV$ and  $\cV_i=\cV\setminus\{i\}$, which is the set of nodes excluding node $i$.

\subsection{Centrality Measures}
\label{sec_centrality}
Depending on the specific application, different notions of centrality may be of interest. The most intuitive notion is arguably the degree of nodes, that is the number of neighbors of each node.  The out-degree of node $i$ is defined by
\bee
D_i=\sum_{j\in\cV} a_{ij}.
\ene

Due to its local nature, this centrality measure is of limited use, and may fail to
capture the actual role of the node in the network. Hence, more refined definitions have been proposed. Classical measures include closeness, betweenness, and eigenvector centralities.

In strongly connected\footnote{When a graph is not strongly connected, a widespread idea is to use the sum of reciprocal of distances, instead of the reciprocal of the sum of distances, with the convention $1/\infty=0$.} networks, a node is considered more important if it is closer to other nodes in the network. While fairness of a node is the sum of its shortest distances to all other nodes, closeness was defined in \cite{bavelas1950communication} as the reciprocal of the fairness, i.e.,
\bee
C_i=\frac{1}{\sum_{j\in\cV}d(i,j)}, \label{closeness}
\ene
where $d(i,j)$ denotes the shortest distance of the forward path from node $i$ to node $j$.

The betweenness \cite{freeman1977set} quantifies the ``control'' of a node on the communication between other nodes, and ranks a node higher if it belongs to more shortest paths between other two nodes in the network. Formally,  the betweenness of node $i$ is defined by
\bee
B_i=\sum_{j, k\in\cV_i}\frac{\sigma(j,k,i)}{\sigma(j,k)},
\ene
where $\sigma(j,k)$ denotes the number of shortest paths from node $j$ to $k$, and $\sigma(j,k,i)$ denotes the number of shortest paths from node $j$ to $k$ containing the node $i$.

Eigenvector centrality measures the influence of a node in a network by the entries of the principal eigenvector of the adjacency matrix, i.e., the score of eigenvector centrality of vertex $i$ is defined as
\bee
x_i= \frac{1}{\lambda} \sum_{j \in \cV} a_{ij}x_j
\ene
where  $\lambda$ is the leading eigenvalue of $A$ \cite{newman2010networks}.   \textcolor[rgb]{0,0,1}{If $A$ is primitive, it follows from Perron-Frobenius theorem~that all the elements of $x$ are positive \cite{horn1985ma}.}

The well-known PageRank is a special case of this centrality measure when the adjacency matrix is suitably redefined as a column stochastic matrix, often called the hyperlink matrix. This notion is discussed in detail in the next subsection.

\subsection{The PageRank Problem}
\label{sec_tprpf}
The PageRank of the Internet is a typical application of the eigenvector centrality for ranking webpages, and has recently attracted the attention of the systems and control community \cite{ishii2010distributed}. In this subsection, we adopt this formulation for the PageRank problem of static graphs, and extend it to temporal networks of time-varying graphs, which is based on the concept of {\em persistent} graph.

\subsubsection{The PageRank Problem in Static Graphs}
The basic idea in ranking pages in terms of the eigenvector centrality is that a page having links from important pages is also important \cite{ishii2010distributed}. This is realized by determining the importance value (eigenvector centrality of a weighted adjacency matrix) of a page as the sum of the contributions from all pages that have links to it. In particular, the importance value $x_i\in[0,1]$ of page $i$ is defined as
\bee\label{pagerank}
x_i=\sum_{j\in\cL_i^1}\frac{x_j}{D_j}, \forall i\in\cV.
\ene
It is customary to normalize the importance values so that $\sum_{i\in\cV}x_i=1$. In view of (\ref{pagerank}), the {\em hyperlink matrix} $W:=(w_{ij})\in\bR^{N\times N}$ is defined by
 \bee\label{hylmatrix}
 w_{ij}:=\left\{\begin{array}{ll} {1}/{D_j}&\text{if}~j\in\cL_i^1,\\0 &\text{otherwise.}\end{array}\right.
 \ene

 \bee
 \label{hyperlink}
 x=Wx, \bone^Tx=1,~\text{and}~x\in[0,1]^N.
 \ene

If $\mathcal{G}$ is not strongly connected, then $W$ may be reducible and $x$ is not unique.  To get around this problem, we consider a {\em random surfer} model \cite{brin1998anatomy}.  Particularly, a surfer may follow the hyperlink structure
of the Web or randomly jumps to other pages with equal probability (which is denoted as ``teleportation" in \cite{ishii2010distributed}). To accommodate this behavior, we adopt a modified hyperlink matrix $M\in\bR^{N\times N}$, which is a convex combination of two column stochastic matrices, i.e.,
 \bee
  \label{hyperlinkm}
 M:=(1-m) W+\frac{m}{N}\bone\bone^{T},
 \ene
 where $m$ is a parameter such that $m\in(0,1)$, and denotes the probability of restarting the random surfer at a given step.  This assures that: (a)  any node is  reachable from all the other nodes, so that the graph is strongly connected; and (b) the Markov chain generated by $M$ is aperiodic and irreducible \cite{norris1998markov}.   Usually, the value $m=0.15$ is chosen at Google \cite{brin1998anatomy}.

In line of \cite{ishii2014PageRank}, equation (\ref{hyperlink}) with $W$ replaced by $M$ is referred to as the {\em PageRank equation}. Note that $M$ is a column stochastic matrix with all positive entries. By Perron Theorem~\cite{horn1985ma}, one is the leading eigenvalue of $M$, and has algebraic multiplicity equal to one. Moreover, the principal eigenvector of $M$ is unique (within a multiplier), and  has all positive elements. This implies that the PageRank equation is written as
  \bee
 \label{pre}
 x=Mx~\text{and}~\bone^Tx=1,
 \ene
 and admits a unique solution $x\in[0,1]^N$. Clearly, $x_i$ also defines the eigenvector centrality of node $i$ in the Internet.

 We use the following example to illustrate the above centrality measures.
 \begin{exmp}
 Consider a directed graph with six nodes in Fig.~\ref{fig_graph}. The adjacency matrix and hyperlink matrix are given by
 \bee
 A=\begin{bmatrix}0&1&0&0&0&0\\
 1&0&1&0&0&0\\
 0&1&0&1&0&0\\
 1&0&1&0&1&1\\
 0&0&0&0&0&1\\
 0&0&1&1&0&0
  \end{bmatrix},~
  W=\begin{bmatrix}0&1\over 2&0&0&0&0\\
 1\over 2&0&1\over 3&0&0&0\\
 0&1\over 2&0&1\over 2&0&0\\
 1\over 2&0&1\over 3&0&1&1\over 2\\
 0&0&0&0&0&1\over 2\\
 0&0&1\over 3&1\over 2&0&0
  \end{bmatrix}~.\nonumber
 \ene
 \end{exmp}
 \begin{figure}
\centering
  \includegraphics[width=7cm]{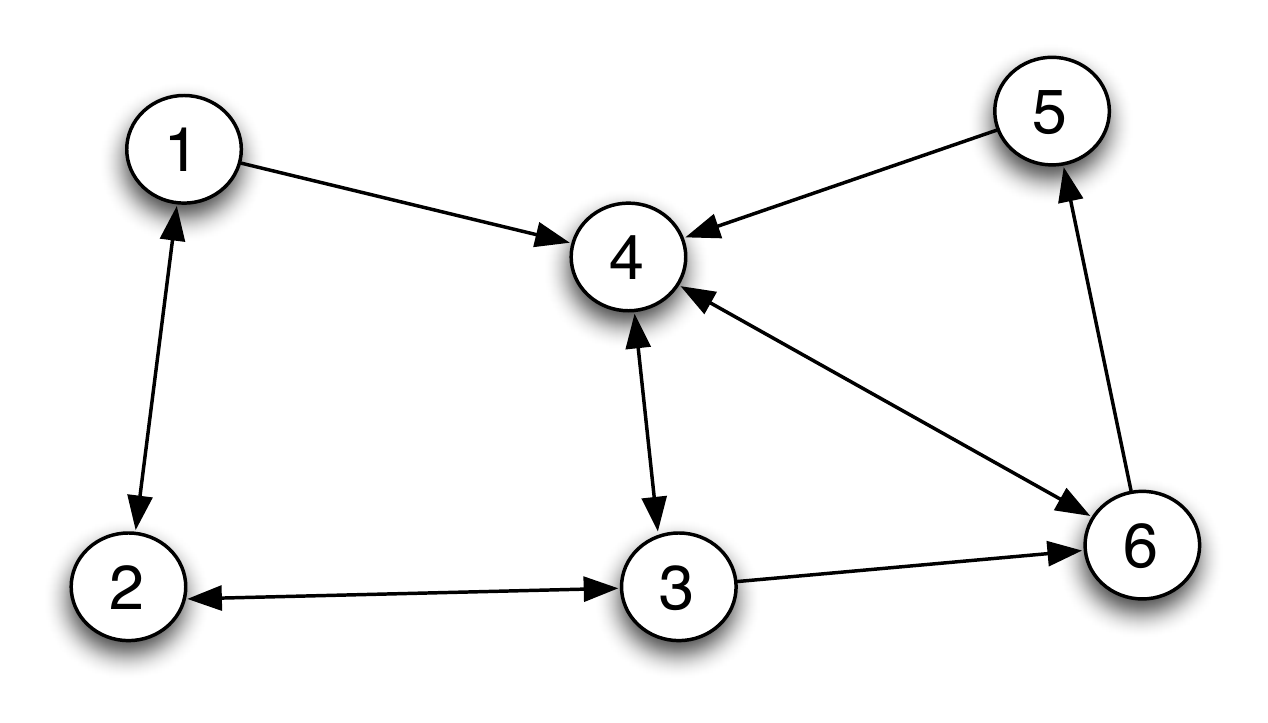}
  \caption{A directed graph with six nodes.}
  \label{fig_graph}
\end{figure}
For comparison, the degree, closeness and betweenness centrality measures are normalized and the normalized version is provided in Table \ref{table_centrality},
from which we may conclude that the nodes $3$ and $4$ are two most important nodes, while nodes $1$ and $5$ are the two least important ones.
\begin{table}
\caption{Centrality measures for the graph in Fig.~\ref{fig_graph} }
\label{table_centrality}
\centering
\begin{tabular}{|c|c|c|c|c|c|c|}
\hline
CENTRALITY & 1 & 2 & 3 &  4 &5 &6  \\
\hline
Degree &  .1667& .1667& .2500 & .1667&0.0833 & .1667 \\
\hline
Closeness &  .1708  &  .1708  &  .2196 &  .1708  & .1281  &  .1398 \\
\hline
Betweenness & .0217 & .1957 & .1957&.4130&0&.1739\\
\hline
PageRank &  .0727 & .1122 &.1986 &.2963 &.1131 &.2072\\
\hline
\end{tabular}
\end{table}

 \subsubsection{The PageRank Problem for Temporal Networks}

Obviously, the network of webpages is time-varying due to the creation or deletion of hyperlinks. In this paper, we consider the case where hyperlinks between
webpages are time-varying and the number of nodes is constant,  i.e., $\cG(k)=\{\cV, \cE(k)\}$, where $\cE(k)$ denotes the set of hyperlinks at time $k$. Then, the hyperlink matrix is no longer fixed and varies over time. An interesting problem is how to define the PageRank in this case.

Intuitively, a spamming  page is a webpage that has too many outgoing hyperlinks in a short amount of time. From this perspective, the spamming pages have only a significant, but short and negative effect on the network. However, a reasonable PageRank definition should not be much affected by spamming pages, and the importance value of a spamming page should be as small as possible. To formalize this observation we introduce a {\em persistent} graph where the effect of spamming (transient) links will be eventually excluded, and only the persistent links may significantly affect the PageRank value.
Specifically,  we define a persistent hyperlink matrix as
\bee\label{perhylkm}
\overline{W}=\lim_{k\rainfty}\frac{\varrho^k}{1+\cdots+\varrho^k}\sum_{t=1}^k \varrho^{-t}  W(t)
\ene
where $\varrho\in (0, 1]$ is a forgetting factor and provides
larger weights on the more recent hyperlink matrices\footnote{Here we implicitly assume that the limit exists.}. Then, the modified link matrix is given by
 \bee
\overline{M}=(1-m) \overline{W}+\frac{m}{N}\bone\bone^{T}.
\ene
Clearly, this definition includes the case of static graphs, and does not require the stationarity of the network. Based on previous discussions, $\overline{M}$ is a column stochastic matrix with positive entries, and the PageRank equation is expressed as
\bee
 \label{hyperlinkvaring}
 \overline{x}=\overline{M}\overline{x}~\text{and}~\bone^T\overline{x}=1.
 \ene
 \section{Distributed Computation of Degree, Closeness and Betweenness Centralities}
\label{sec_dcdcbc}
In this section, we focus on the distributed computation of the degree, closeness and between centralities in the network $\cG=(\cV,\cE)$. Obviously, the computation of the degree centrality is trivial. The key technique for the closeness and betweenness centrality measures lies on the computation of the shortest distance for each pair of nodes in the graph. This problem is solvable in a centralized manner by a linear program\cite{bertsekas1998network}. Here we propose a distributed algorithm to incrementally compute the shortest path between two nodes, which is then used to compute the closeness centrality of a directed graph and the betweenness centrality (the latter in the special case of an oriented tree).

\subsection{Degree and Closeness Computation}
Let $d_{\max}=\max_{i,j\in\cV}d(i,j)$, which is also obtainable in a distributed manner  \cite{garin2012distributed}. \textcolor[rgb]{0,0,1}{Then, $\cR_i^t=\emptyset$ for all $t>d_{\max}$ and}

\bee
d(i,j)=t, \forall j\in\cR_i^t \label{distance}
\ene
and
\bee
\bigcup_{t=1}^{d_{\max}}\cR_i^t\subseteq \cV_i.\label{partition}
\ene

If $\cG$ is strongly connected, the equality holds. Otherwise,
there exists a node that can not be reached from node $i$ and the left hand side of (\ref{partition}) is a strict proper subset of $\cV_i$.
It should be noted that $\cR_i^t$ can be empty, and $\cR_i^j\cap \cR_i^k=\emptyset$ if $k\neq j$. By (\ref{distance}) and (\ref{partition}), the computation of $d(i,j), j\in\cV_i$ is the same as the task of partitioning the set $\cV_i$.  By the local (one-hop) interaction, the $(t+1)$-hop \textcolor[rgb]{0,0,1}{neighbor set} of node $i$ is recursively given by
\bee
\cR_i^{t+1}=\bigcup_{j\in\cR_i^1} \cR_j^{t} -\bigcup_{k=1}^t \cR_i^k.
\ene

That is, the one-hop neighbors of node $i$ send their $t$-hop neighbors to node $i$, and node $i$ checks  whether they belong to its $k$-hop neighbors, $k\le t$. If not, then this node has a minimum distance $t+1$ to node $i$.   Obviously, this algorithm relies only on the local interaction with one-hop neighbors and it is provided at the end of this section.

 \subsection{Betweenness Computation of Trees}

The betweenness centrality is essential in the analysis of social networks, but costly to compute for a generic graph. A space and time efficient centralized algorithm has been proposed in \cite{brandes2001faster}. In this section, we propose a distributed method to compute the betweenness centrality measure for an oriented tree, which is a directed acyclic graph \cite{godsil2001algebraic}.

Note that a tree does not contain any cycle, and has the key feature that the number of shortest paths between two nodes is always equal to one. This implies that
 \bee
B_i=\sum_{j, k\in\cV_i}\sigma(j,k,i).
\ene

To compute $B_i$, define the set of all reachable nodes from node $i$ by
\bee
\cR_i=\bigcup_{t=1}^{d_{\max}} \cR_i^t.
\ene
We also denote the set of reachable nodes from node $i$ via the directed link $(i,j)\in\cE$ by
$\cR_{i\ra j}$. By convention, if there is no edge from node $i$ to $j$, we set $\cR_{i\ra j}=\emptyset$. Since $\cG$ is an oriented tree, then   \textcolor[rgb]{0,0,1}{for any $k\neq j$,  $\cR_{i\ra j}\cap \cR_{i\ra k}=\emptyset$. To elaborate it, suppose that} there exists a node $v\in \cR_{i\ra j}\cap \cR_{i\ra k}$. Then, there exist two directed paths from node $i$ to node $v$, respectively, via node $j$ and node $k$. If we replace directed edges with undirected ones, we obtain a cycle starting from node $i$ via node $v$ to node $i$. This is in contradiction to the definition of an oriented tree.
Hence, it is clear that
$$\cR_i=\bigcup_{j\in\cV}\cR_{i\ra j}.$$

Similarly,  define the set of all nodes linking to node $i$ by $\cL_i$
and the set of nodes linking to node $i$ via the directed link $(j, i)\in\cE$ by
$\cL_{j\ra i}$, which again satisfies that
$$\cL_i=\bigcup_{j\in\cV}\cL_{j \ra i}.$$

Now, we show how to compute the betweenness of node $i$. If there exists $u$ and $v$ such that $\sigma(u,v,i)= 1$,  we can always find a pair of nodes $j$ and $k \neq j$ such that $u\in \cL_{j\ra i}$ and $v\in \cR_{i\ra k}$. Conversely, for any pair of nodes $u\in \cL_{j\ra i}$ and $v\in \cR_{i\ra k}$, we obtain that $\sigma(u,v,i)=1$. Otherwise, it contradicts the definition of an oriented tree. This implies that betweenness of node $i$ can be explicitly computed by
\bee
B_i=\sum_{j,k\in\cV_i, j\neq k} |\cR_{i\ra j}|\cdot |\cL_{k\ra i}|.
\ene

In summary, the distributed computation of the degree, closeness and betweenness  centrality measures are given in Algorithm~\ref{alg_closeness}. We remark that the closeness and betweenness centralities are computed for an undirected tree from a dynamical system point of view in \cite{wang2014distributed,wang2013distributed}.
\begin{algorithm}
\label{alg_closeness}
\caption{Distributed computation for degree, closeness and betweenness centralities}
\begin{itemize}
\item Initialization: for every $i\in\cV$, compute $\cR_i^1$ and $\cL_i^1$ by local interactions;
\item For $t < d_{\max}$.  Given $\cR_j^t$ and $\cL_j^t$, which are obtained from local interactions with $j\in \cR_i^1$ and $j\in \cL_i^1$, respectively. Compute
\beq
\cR_i^{t+1}&=&(\cup_{j\in\cR_i^1} \cR_j^{t})-\cup_{k=1}^t \cR_i^k;\nonumber\\
\cL_i^{t+1}&=&(\cup_{j\in\cL_i^1} \cL_j^{t})-\cup_{k=1}^t \cL_i^k;\nonumber
\enq
\item Compute the cardinality of $\cR_{i\ra j}$ and $\cL_{j\ra i}$ by
$$
|\cR_{i\ra j}|=\sum_{t=1}^{d_{\max}}|\cR_j^{t} |~\text{and}~|\cL_{j\ra i}|=\sum_{t=1}^{d_{\max}}|\cL_j^{t} |.
$$
Then, output the degree, closeness and betweenness centralities by
\beq
D_i&=&|\cR_i^{1}|; \nonumber \\
C_i&=&1/\sum_{t=1}^{d_{\max}}(|\cR_i^{t} |\cdot t); \nonumber\\
B_i&=&\sum_{j\in\cR_i^1,k\in\cL_i^1, j\neq k} |\cR_{i\ra j}|\cdot |\cL_{k\ra i}|.\nonumber
\enq
\end{itemize}
\end{algorithm}

\begin{rem}[Computation Complexity] In Algorithm~\ref{alg_closeness}, the number of operations in each node is $O(|\cV|\cdot |\cR_{\max}^1|)$ for every $t$ where $|\cR_{\max}^1|=\max_{i\in\cV}|\cR_i^1|$. As $t < d_{\max}$, the total complexity of a single node is $O(|\cV| \cdot |\cR_{\max}^1|\cdot d_{\max})$, which is smaller than the centralized Brandes' algorithm $O(|\cV|\cdot |\cE|)$. However,  the total number of operations in the whole network is $O(|\cV|^2 \cdot |\cR_{\max}^1|\cdot d_{\max})$, which is larger than Brandes' algorithm due to the limited access of information in computation.
\end{rem}

 \section{Distributed Computation of PageRank}
 \label{sec_dcp}
 In this section, we provide incremental algorithms to distributedly compute the PageRank, which is also a special case of  eigenvector centrality, under two scenarios depending on whether an individual node  has the knowledge of the network size. We reformulate the PageRank computation as a least squares problem, and propose a new type of randomized Kaczmarz algorithm to incrementally compute the PageRank.  The essential idea lies in the integration of the randomized incremental algorithm with a random surfer model. The striking features of our algorithm are at least threefold: (a) It can be conveniently implemented in a {\em fully distributed} manner by using only local information of an individual page. (b) It is based on {\em Markovian randomization} (instead of simpler i.i.d. randomization), which accounts for the Web structure quite well.  (c) It can be simply generalized to accommodate  {\em temporal networks}, as discussed in Section \ref{sec_tvg}. Further comments and connections with the Kaczmarz algorithm are given in Remark~\ref{Kaczmarz} in Section V.

 \subsection{Least Squares Reformulation of the PageRank Problem}
Now, we reformulate the problem of solving the PageRank equation (\ref{pre}) as a least squares problem by using the following result.
\begin{lem}[Equivalent Equation for PageRank] The solution to (\ref{pre}) is equivalent to that of the following equation
 \bee
(I-(1-m)W) x=\frac{m}{N}\bone. \label{pagerankvector}
 \ene
\end{lem}
\begin{proof}
By $\bone^T x=1$, then $Mx=(1-m)Wx+\frac{m}{N}\bone=x$, which implies (\ref{pagerankvector}). Conversely, it follows from (\ref{hylmatrix}) that $W$ is a column stochastic matrix. This implies that the eigenvalue with the largest magnitude of $W$ is one. Since $m$ strictly belongs to $(0, 1)$, we obtain  that $I-(1-m)W$ is nonsingular, and its spectral radius is strictly less than one. Then, it follows from (\ref{pagerankvector})  that
\begin{eqnarray*}
x&=&\frac{N}{m} (I-(1-m)W) ^{-1}\bone \\
&=&\frac{N}{m} \sum_{i=0}^{\infty}(1-m)^iW^i  \bone.
\end{eqnarray*}

By (\ref{hylmatrix}),  it is clear that $\bone^T W=\bone^T$. This implies that
$$
\bone^Tx=\frac{N}{m} \sum_{i=0}^{\infty}(1-m)^i \bone^T W^i\bone =\frac{1}{m} \sum_{i=0}^{\infty}(1-m)^i =1.
$$
That is, the solution to (\ref{pagerankvector}) also solves (\ref{pre}).
\end{proof}

The advantage of the PageRank equation in (\ref{pagerankvector}) over (\ref{pre}) is to drop the normalization constraint $\bone^T x =1$, which implies that the PageRank can be solved via an {\em unconstrained} optimization as follows
 \bee\label{optimization}
 x=\argmin_{x\in\bR^N} \twon{(I-(1-m)W)x-\frac{m}{N}\bone}^2.
 \ene

 Then, the key problem is how to efficiently compute $x$ in a distributed manner. Let $W_i$ be the $i$-th row of $W$,  $H_i=e_i-(1-m)W_i$, where $e_i$ is the $i$-th row of an identity matrix, and $y_i=m/N$. The optimization problem  (\ref{optimization}) is easily rewritten as a LS problem
  \bee\label{leastsquare}
 x=\argmin_{x\in\bR^N}\sum_{i=1}^N (y_i-H_ix)^2.
 \ene
Note from (\ref{hylmatrix}) that $H_i$ is computable by using only information from the neighbors with outgoing links to node $i$, which is obviously known to node $i$.  Usually, $H$ is a sparse matrix, and contains many zero entries.

 Similarly, the LS reformulation of the PageRank problem for temporal networks is given as
 \bee\label{leastsquarevarying}
 \overline{x}^{ls}=\argmin_{x\in\bR^N}\sum_{i=1}^N (y_i-\overline{H}_ix)^2,
 \ene
where $\overline{H}_i=e_i-(1-m)\overline{W}_i$ and $\overline{W}_i$ is the $i$-th row of $\overline{W}$.

 Next, we state a result which guarantees the solvability and uniqueness of the LS problems (\ref{optimization}) and (\ref{leastsquarevarying}).
 \begin{lem}\label{lem_definite} For the PageRank computation,
$ \sum_{i=1}^N H_i^TH_i$ in (\ref{leastsquare}) is positive definite.
 \end{lem}
 \begin{proof}
Note that $I-(1-m)W$ is non-singular. The positive definiteness of $\sum_{i=1}^N H_i^TH_i$ follows from
$$\sum_{i=1}^N H_i^TH_i=(I-(1-m)W)^T(I-(1-m)W)>0.$$
Thus, the proof is completed.
  \end{proof}

 By using standard result on LS techniques \cite{kailath2000le}, it is obvious that the solution to (\ref{leastsquare}) is exactly expressed as
\bee
x=\big(\sum_{i=1}^NH_i^TH_i\big)^{-1}\big(\sum_{i=1}^NH_i^Ty_i\big).\label{lsec}
\ene

Obviously, the same arguments continue to hold for the time-varying graphs by directly substituting $H$ with $\overline{H}$.

 \subsection{Incremental Algorithms with Known Network Size}
 \label{sec_ia}
To obtain the LS solution, the formula in (\ref{lsec}) requires to utilize all $y_i$ and $H_i$ for computing the inverse of a square matrix of order $N$, which is in fact the network size.  Even worse, the sparsity of $H_i$ is not used as the matrix $\sum_{i=1}^NH_i^TH_i$ does not preserve a sparsity structure. Due to the size of the network, the computational cost of  inverting the matrix in (\ref{lsec}) is very large. This motivates to design incremental algorithms to compute the PageRank $x$ and $\overline{x}$.

Now we propose a randomized incremental algorithm with a {\em diffusion} vector $x(k)\in\bR^N$ (shown in equation (\ref{diffusion}) below) to compute the PageRank. At iteration $k$, a node, indexed as $s(k) \in\cV$, is randomly
selected according to the method described in Section V-A. This page incrementally updates $x(k)$ by performing a fusion algorithm
\beq
x(k+1)&=&x(k)-\frac{1}{2N}\cdot  \frac{d (y_{s(k)}-H_{s(k)}x)^2}{dx}|_{x=x(k)}\nonumber \\
&=&x(k)+\frac{1}{N}H_{s(k)}^T(y_{s(k)}-H_{s(k)}x(k)),\label{diffusion}
\enq
and the initial condition $x(0)=0$.

In comparison with the traditional gradient algorithms \cite{boyd2004convex}, the above algorithm only guarantees that a component $(y_{s(k)}-H_{s(k)}x(k))^2$ of the total cost $\sum_{i=1}^N (y_i-H_ix(k))^2$ is improved. Since it does not take other components into consideration, the random process $s(k)$ should be carefully designed to reduce the total cost. Thus, a deeper investigation, performed in Section V, is required to prove that the iteration in (\ref{diffusion}) solves the LS problem (\ref{leastsquare}).
\begin{rem}[Computation Complexity]
Due to the sparsity of $H_{s(k)}$, the number of computations in (\ref{diffusion1}) is small, e.g., the number of multiplications is only twice of the in-degree of node $s(k)$ and scalable to the network size. Thus, it can be easily performed by a cheap processor.
 \end{rem}

\subsection{Distributed Implementation with Unknown Network Size}
\label{sec_diuns}
Generally, the network size $N$, which is a global parameter, is unknown to an individual node.  In this case, there are two critical issues for running the fusion algorithm (\ref{diffusion}). The first issue is that both the stepsize $1/N$ and $y_{s(k)}=m/N$ are not available. To resolve this problem, we modify the incremental algorithm as
\bee
x(k+1)=x(k)+\alpha(k)\cdot H_{s(k)}^T(\widehat{y}_{s(k)}-H_{s(k)}x(k)),\label{diffusion1}
\ene
where $\widehat{y}_{s(k)}=m\cdot\alpha(k)$, and the stepsize $\alpha(k)$ is given by
\bee
\alpha(k)=\frac{1}{k+1}\sum_{t=0}^{k}\chi_{s(k)}(s(t))\label{stepsize}
\ene with the standard indicator function $\chi_{s(k)}(\cdot)$ being defined as $\chi_{s(k)}(x)=1$ if $x=s(k)$ and $0$ otherwise.

It is easy to check that $\alpha(k)$ counts the frequency of randomization of page $s(k)$.  Under relatively mild conditions,
in Section V we show that $\alpha(k)^{-1}$ converges to $N$ with probability one as the number of updates $k$ tends to infinity. Roughly speaking, this implies that the measurement estimate $\widehat{y}_{s(k)}$ asymptotically converges to $y_{s(k)}=m/N$.

 \begin{rem} [Clock Free Algorithm] The total number of activations $k+1$ for $\alpha(k)$ can be locally obtained by the activated node via a local counter. In particular, the counter increases by one at each activation, after which it is passed to (detected by) the next activated node. Under this mechanism, the activated node is able to access the total number of activations, and does not require a clock in (\ref{diffusion1}).  This is contrary to the randomized algorithms proposed in \cite{ishii2014PageRank} which indeed requires a global clock.
 \end{rem}

The second critical issue regarding unknown network size is that the dimension of $x(k)$ is unknown to an individual node. This implies that the iteration (\ref{diffusion1}) can not be executed locally either.  However, this is not a problem as only a coordinate block of $x(k)$ in (\ref{diffusion1}) is updated per iteration, and does not require any node to know the network size. To implement (\ref{diffusion1}) in a distributed manner, each page is in charge of computing a
 ``sub-PageRank'' for its neighboring pages and itself. Specifically, let $x(k)=[x_1(k),\ldots,x_N(k)]^T$ and
 $$x^{(i)}(k)=[x_i(k),x_j(k),j\in\cL_i^1]^T,$$
where  $x^{(i)}(k)$ consists of all the estimated importance values of the neighbors of node $i$ and itself. By sequentially sorting out all nonzero elements of $H_i$ and collecting them into a new vector $\widehat{H}_i\in\bR^{|\cL_i^1|+1}$, whose dimension is much smaller than $N$ for a sparse network, it follows from (\ref{diffusion1}) that
\beq
x^{(s(k))}(k+1)&=&x^{(s(k))}(k)+\label{condensed}\\
&&\alpha(k)\cdot \widehat{H}_{s(k)}^T(\widehat{y}_{s(k)}-\widehat{H}_{s(k)}x^{(s(k))}(k)).\nonumber
\enq

This iteration is fully localized, see Fig.~\ref{fig_conf} for illustration. In particular, the stepsize $\alpha(k)$ counts the percentage of updates that have been completed in page $s(k)$, which is inherently known to this page without any global information. Similarly, $\widehat{H}_{s(k)}$ is solely decided by the incoming links to page $s(k)$, which is again known to this page. It is also consistent with the observation that every page is only concerned with the rank of neighboring pages, and returns a sub-PageRank. Subsequently, each neighboring page detects its updated value in the sub-PageRank from page $s(k)$.

\begin{figure}
\centering
 \includegraphics[width=8.5cm]{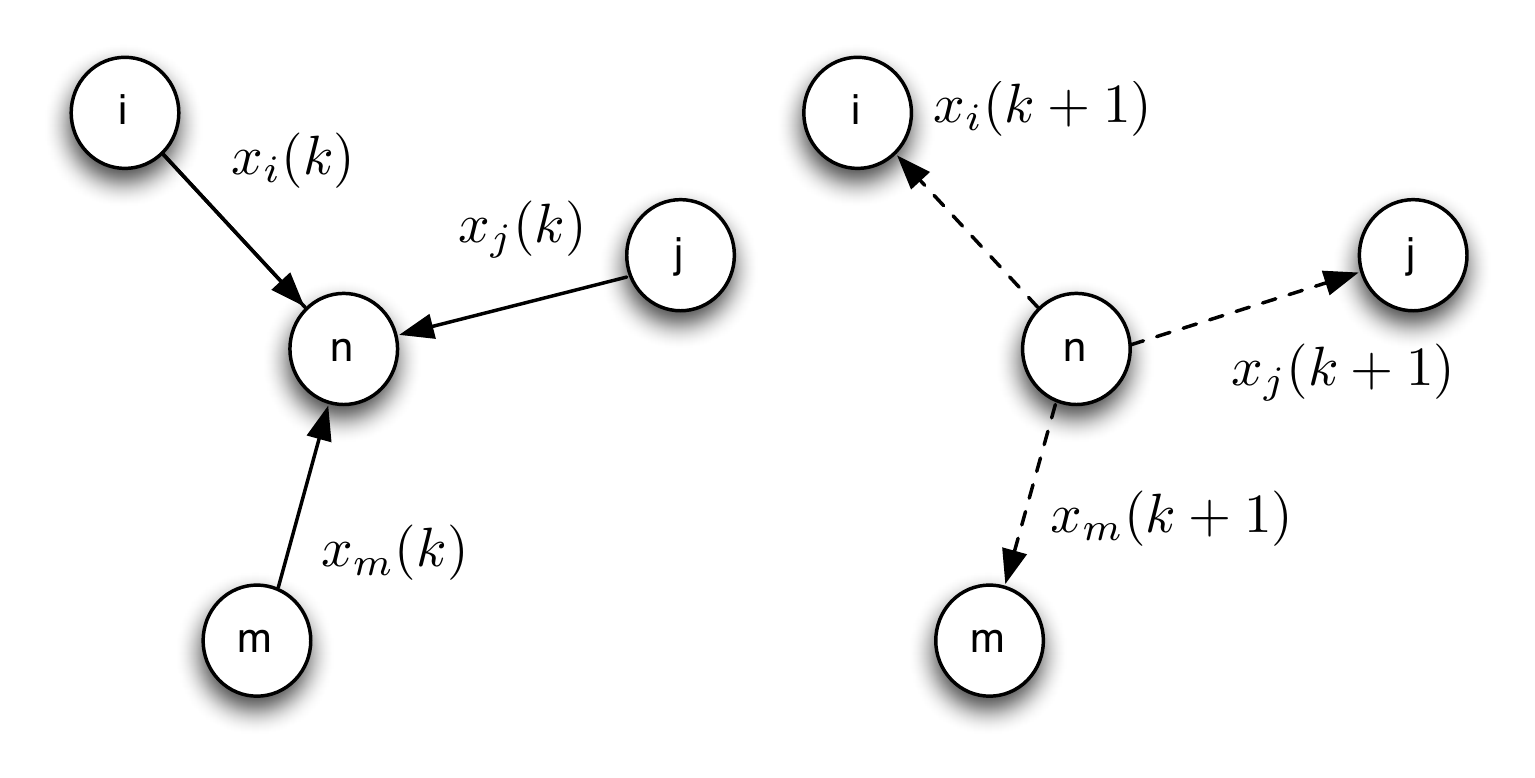}
  \caption{Local PageRank computation. If node $n$ is initiated at iteration $k$, which is detected by its neighbors, every neighbor sends its importance value $x_j(k), j\in\cL_n^1$ to this node for a local computation. Node $n$ assembles their values and performs the iteration (\ref{condensed}). The updated values are again detected by the neighbors.}
  \label{fig_conf}
\end{figure}

In summary,  the fusion algorithm (\ref{condensed}) can also be implemented in a fully distributed way for networks with unknown size. Thus, the remaining problem is to show the convergence of $x(k)$ in (\ref{diffusion1}) or  (\ref{condensed}) to the PageRank $x$. This will be addressed in Section V. Algorithm~\ref{alg_eigenvector} provides the distributed computation of PageRank with unknown network size.

\begin{rem} [Unknown Network Size] It is noteworthy to remark that the existing distributed optimization algorithms \cite{necoara2013random,nedic2015distributed,iutzeler2015explicit,duchi2012dual,wei2013distributed} cannot be directly applied here because they require to know the network size. \end{rem}

\begin{algorithm}
\label{alg_eigenvector}
\caption{Distributed computation of PageRank with unknown network size}
\begin{itemize}
\item Initialization: for every $i\in\cV$, set $x_i(0)=0$ and $s(0)=0$;
\item If $s(k)=i$,  node $j\in\cL_i^1$ sends its importance value $x_j(k)$ to this node for a local computation as in (\ref{condensed}).  Node $j\in\cL_i^1\cup \{i\} $ updates its importance value from $x_j(k)$ to $x_j(k+1)$;
\item Repeat.
\end{itemize}
\end{algorithm}

\section{Randomization and Convergence of the PageRank  for Static Graphs}
\label{sec_rce}

In this section, we analyze randomized algorithms, which may be superior to deterministic algorithms
in many aspects \cite{tempo2013randomized}, for the PageRank computation of static graphs
with both known and unknown network size. Then, we provide convergence results almost surely and in the sense of $L^p (1< p\le 2)$. The diffusion algorithm studied here is based on a specific construction of the random process $s(k)$ in (\ref{diffusion1}) which reduces the total cost in (\ref{leastsquare}), and drives the iteration to converge to the LS solution.

\subsection{Randomization of $s(k)$}

First, motivated by the random surfer model described in Section \ref{sec_tprpf}, we design the random process $s(k)$ which dictates how the nodes of the Web are selected.
At time $k$, a random surfer has a prior vector $x(k)$, and randomly browses a page, indexed as $s(k)$. After inspecting the incoming links of this page and the number of visits to this page (to update the stepsize $\alpha(k)$), the surfer incrementally updates this vector by (\ref{diffusion}) to $x(k+1)$. If $s(k)=i$, the surfer randomly jumps to page $j$ ($j\neq i$) either by following the hyperlink structure with probability $(1-\omega)P_{ij}^W$, where $\omega \in [0,1]$ is a parameter and
$P_{ij}^W$ is a probability that is defined based on the Web structure, or performs a random jump to page $j$ with probability $\omega/N$.  Here $j=i$ means that the surfer refreshes the current page with probability $(1-\omega)P^W_{ii}$ or randomly returns to the current page with probability $\omega/N$.

Motivated by this discussion, we are now ready to define the Markov chain, and its transition probability matrix, of the random process $s(k)$ in (\ref{diffusion1}).

\begin{defi}[Transition Probability Matrix for Static Graphs]\label{assumption1}
The incremental process $s(k)$ is a Markov chain with a transition probability matrix $P=(1-\omega)P^W+\frac{\omega}{N}\bone\bone^T$ where $\omega\in[0,1]$, and $P^W$ is given by
\bee
P^W_{ij}=
\left\{
\begin{array}{ll}
\min\{\frac{1}{D_i+1},\frac{1}{D_j+1}\}&\text{if}~(i,j)\in\cE,\\
1-\sum_{(i,k)\in\cE}P^W_{ik}&\text{if}~i=j,\\
0&\text{otherwise.}
\end{array}
\right.\nonumber
\ene
\end{defi}

The matrix $P^W$ (generally called the Metropolis-Hastings matrix) is {\em doubly stochastic}, and it
 coincides a so-called min-equal neighbor scheme \cite{xiao2007distributed}. We notice that the second term $\bone\bone^T/N$ in the transition probability matrix is motivated by the teleportation model
described in Section \ref{sec_tprpf}. The convex combination parameter $\omega\in[0,1]$ represents two modes of randomization, which are studied in a unified setting: (a) centralized randomization of (\ref{diffusion}) where $\omega$ can be selected as any value in $(0,1]$. The extreme case $\omega=1$ means that the random process $s(k)$ reduces to an i.i.d. process as in the randomized Kaczmarz algorithm \cite{zouzias2013randomized}. Usually, a larger value of $\omega$ implies better convergence performance of (\ref{diffusion}).
(b) distributed randomization of (\ref{condensed}) in Section \ref{sec_diuns}, where randomization is allowed to pass from one page only to its neighbors and the network size is unknown.

The case (b) corresponds to $\omega=0$ in the transition probability matrix $P$ and is applicable only to a strongly connected Web. \textcolor[rgb]{0,0,1}{Suppose the Web is not strongly connected, there exists a node $i$ that can not be reached from some other node $j$. Then, once node $j$ is activated, node $i$ can {\em never} be activated again.} Even though the Web is generally not strongly connected, strong connectivity may be enforced in practice, as discussed in \cite{ishii2014PageRank}. In particular, if the network contains dangling  pages and is not strongly connected, the surfer may use a back button and return to the previously visited page to continue randomization. In the centralized randomization case, the second term of $P$ requires that every pair of pages should be reachable from each other. This is a natural assumption, which is now formally stated.

\begin{assum}[Strong connectivity]\label{strg} For the distributed randomization,  i.e., $\omega=0$ in the transition probability matrix $P=(1-\omega)P^W+\frac{\omega}{N}\bone\bone^T = P^W$, we assume that the network of webpages is strongly connected.
\end{assum}

A nice property, which follows from Definition~\ref{assumption1} of transition probability matrix and Assumption \ref{strg} on strong connectivity is that $s(k)$ admits a unique stationary distribution, which is  uniformly distributed over the set $\cV$. This implies that
\beq
\cH&:=&\lim_{k\rainfty}\bE[H^T_{s(k)} H_{s(k)}]=\frac{1}{N}\sum_{i=1}^N H_i^TH_i,\label{meanh}\\
\cY&:=&\lim_{k\rainfty}\bE[H^T_{s(k)} y_{s(k)}]=\frac{1}{N}\sum_{i=1}^N H_i^Ty_i,\nonumber
\enq
where the expectation operator $\bE[\cdot]$ is taken with respect to the random process $s(k)$.

\subsection{Convergence Analysis}

The first convergence result we state for the randomized incremental algorithm is for known network size.

\begin{thm} [Convergence with Known Network Size] \label{thm_convergence}  Under Assumption~\ref{strg} on strong connectivity, the randomized incremental algorithm (\ref{diffusion}) with transition probability matrix given in Definition~\ref{assumption1} enjoys the following properties:
\begin{enumerate}
\renewcommand{\labelenumi}{\rm(\alph{enumi})}
\item There exist $k_0>0$, \textcolor[rgb]{0,0,1}{$M_1>0$, and $\rho_1\in(0,1)$}  such that $\twon{x(k)-x}\le M_1\cdot \rho_1^k$ with probability one for \textcolor[rgb]{0,0,1}{all} $k>k_0$.
\item   Given any $p\in[1,2]$, \textcolor[rgb]{0,0,1}{there exist $M_2>0$ and $\rho_2\in(0,1)$ such that $\twon{x(k)-x}_{L^p}\le M_2\cdot\rho_2^k$ for all $k>0$}.
 \end{enumerate}
\end{thm}

\begin{rem} [Convergence with Known Network Size]
If the network size is known,  the randomized algorithm (\ref{diffusion}) exponentially converges to the PageRank  almost surely and in the sense of $L^p (1\le p\le 2)$.  From this point of view, this convergence result is much stronger than those stated for the PageRank algorithms in \cite{ishii2010distributed}, where mean square error convergence is stated with linear convergence rate \textcolor[rgb]{0,0,1}{$O(1/k)$}, see further discussions in Section V-C.
\end{rem}

\begin{rem} [Relations with Randomized Algorithms for Distributed Optimization]
The idea of using randomized incremental algorithms has been adopted in \cite{johansson2009randomized,nedic2001incremental}  for solving distributed optimization problems. However, the convergence in these papers requires the stepsizes decreasing to zero, which inevitably reduces the convergence rate. In addition, the proof of convergence is completely different from that of Theorem~\ref{thm_convergence}.
\end{rem}

\begin{rem}[Comparisons with the Kaczmarz Algorithm]\label{Kaczmarz}
There are several striking differences between the randomized diffusion algorithm (\ref{diffusion}) studied in this section and the celebrated randomized Kaczmarz algorithm in \cite{zouzias2013randomized}, which has been originally developed to solve generic linear equations.

First, the random process $s(k)$ is Markovian and takes into account the surfer's browsing history.
This case covers the i.i.d. process (which is assumed in the randomized Kaczmarz algorithm) as a very special case, i.e. in Definition~\ref{assumption1}, $\omega$ is equal to one in the transition probability matrix $P$.
Notice that Markovian randomization can be easily implemented by only allowing nodes to communicate with their neighbors. Therefore, it is a very general randomization scheme which is suitable to describe a random surfer model on the Web.

Second, the use of Markovian randomization in (\ref{diffusion}) makes the analysis of convergence much more difficult. Clearly, the techniques for studying convergence of the i.i.d. randomized Kaczmarz algorithm are no longer applicable. The same comment holds for random coordinate descent algorithms
 \cite{nesterov}.
In fact, the proof of convergence of Theorem~1 is based on sophisticated technical results in the theory of stochastically time-varying systems \cite{guo1994stability}. More precisely, the key technical result is Lemma~3 stated below, which deals with convergence of transition matrices, and it is based on ergodicity properties of $\phi$-mixing processes.

Third, contrary to the randomized Kaczmarz algorithm, the algorithms and the convergence results can be extended to temporal networks with time-varying links, which is a significant extension, subsequently provided in Section VI.

Finally, we notice that, when applying the randomized Kaczmarz algorithm \cite{zouzias2013randomized} to the PageRank problem, the probability of choosing a node is proportional to the size of the regression vector, i.e.,
\bee
\bP\{s(k)=i\} = \twon{H_i}^2 / \sum_{i=1}^{N} \twon{H_i}^2.\label{prob}\ene
This implies that the probability vector for randomization relies on the global information in the denominator of (\ref{prob}). This is different from the randomized diffusion algorithm studied in this section, where the distribution of $s(k)$ tends to be asymptotically uniform and is unknown to any node.
\end{rem}

The proof of Theorem~\ref{thm_convergence} depends on a key technical result, which is stated in this section for completeness. The proof is given in the Appendix.

\begin{lem}[Convergence of Transition Matrices]\label{lem_trans}
Let the transition matrix be $\Phi(k+1,j)=(I-1/N \cdot H_{s(k)}^TH_{s(k)})\Phi(k,j)$ for all  $j\le k$ and $\Phi(k,k)=I$.
Under Assumption \ref{strg} on strong connectivity, the randomized incremental algorithm (\ref{diffusion}) with transition probability matrix given in Definition \ref{assumption1}  enjoys the following properties:
\begin{enumerate}
\renewcommand{\labelenumi}{\rm(\alph{enumi})}
\item There exists a positive $M>0$ and $\rho\in(0, 1)$ such that
\bee\label{convergencerate}
\twon{\Phi(k,j)}_{L^2}\le M\rho^{k-j}, \forall k \ge j\ge 0.
\ene
\item There exists a sufficiently large $k_0>0$ and $\eta\in(0, 1)$  such that with probability one
\bee
\twon{\Phi(k,j)}\le \eta^{\lfloor\frac{k-j}{N}\rfloor}~\text{if}~ k-j>k_0,
\ene
where $\lfloor x \rfloor$ is the floor function, i.e. the largest integer not greater than $x\in\bR$.
\item  With probability one, it holds that
\bee\label{bound_trans}
\lim_{k\rainfty}\sum_{j=0}^k\twon{\Phi(k,j)}<\infty.
\ene
\end{enumerate}
\end{lem}
{\em Proof of Theorem~\ref{thm_convergence}}:~By (\ref{pagerankvector}) and (\ref{leastsquare}), it is obvious that $y_{s(k)}=H_{s(k)}x$.  Let $e(k):=x(k)-x$, it follows from (\ref{diffusion}) that
$e(k+1)=(I-1/N \cdot H_{s(k)}^TH_{s(k)})e(k)=\textcolor[rgb]{0,0,1}{ \Phi(k+1,0)}e(0).$
The rest of proof is a trivial consequence of Lemma~\ref{lem_trans}. \qed

Next, we establish the convergence results of the randomized incremental algorithms (\ref{diffusion1}) and (\ref{condensed}) with unknown network size.

\begin{thm}[Convergence with Unknown Network Size]\label{thm_convergence1}
Under Assumption \ref{strg} on strong connectivity, the randomized incremental algorithms  (\ref{condensed}) with transition probability matrix given in Definition \ref{assumption1} enjoys the following properties:
\begin{enumerate}
\renewcommand{\labelenumi}{\rm(\alph{enumi})}
\item (Almost sure convergence)
$\lim_{k\rainfty}x(k)=x$ with probability one.
\item ($L^p$ convergence) $\lim_{k\rainfty} \twon{x(k)-x}_{L^p}=0$ for any $ \textcolor[rgb]{0,0,1}{p\ge1}$.
 \end{enumerate}
\end{thm}

The effect of the stepsize $\alpha(k)$ on the randomized incremental algorithms is essentially the same as that of $1/N$ as formally stated in the next result.

\begin{lem} \label{lem_ergodic}
Under Assumption \ref{strg} on strong connectivity, the stepsize with transition probability matrix given in Definition \ref{assumption1} satisfies the property
\bee
\lim_{k\rainfty}\alpha(k)=1/N.
\ene
\end{lem}

\begin{proof} Without loss of generality, we assume that
$s(0)$ has a uniform distribution over $\cV$. Then,  $s(k)$ is an ergodic and irreducible process. By the Ergodic Theorem~\cite{norris1998markov}, it follows that
\bee
\lim_{k\rainfty}\alpha(k)=1/N,
\ene
where the equality is due to the uniform distribution of $s(j)$ for all $j\ge 0$.
\end{proof}

\begin{rem} [Exponential Convergence]
By Theorem~\ref{thm_convergence}, once $\alpha(k)$ is close to $1/N$, the convergence becomes exponential.
\end{rem}
%
%

\begin{rem}[Estimation of the Network Size]
Clearly, the network size $N$ is a global information and is unknown to an individual node.
Distributed estimation of the network size has been addressed in e.g. \cite{varagnolo2014distributed} using various techniques. In fact, the diffusion algorithm of this section can be simply used to locally estimate the network size by using the reverse of the stepsize, i.e., $\lim_{k\rainfty}\alpha^{-1}(k)=N$ with probability one.
\end{rem}

\subsection{Relations to the State-of-the-art}
In \cite{ishii2014PageRank,ishii2010distributed},  the PageRank problem is solved by designing the so-called distributed link matrices. Specifically, every node $i$ is associated to a link matrix $A_i$, whose $i$-th column coincides with the $i$-th column of $W$ of this paper. The distributed update scheme with randomization is of the form
\beq
y(k+1)&=&(1-m)A_{s(k)}y(k)+\frac{m}{N}\cdot {\bf 1},\label{pagerankdis} \\
 x(k+1)&=&\frac{1}{k+1}y(k+1)+\frac{k}{k+1}x(k),\label{pagerankiishi}
\enq
where the initial condition $y(0)=x(0)$ is chosen as any probability vector. Under the assumption that $s(k)$ is an i.i.d. process with a uniform distribution and $m\in(0, 1)$, it is proved that $\lim_{t\rainfty}\bE[\twon{x(t)-x}^2]=0$ with a linear convergence rate in \cite{ishii2010distributed} and  $\lim_{t\rainfty}x(t)=x$ almost surely in \cite{zhao2013convergence,lei2015distributed}.

In comparison, the algorithms of this paper are derived via an optimization approach, which incrementally improves the total cost under Markovian randomization.
While the PageRank algorithm in (\ref{pagerankiishi}) is motivated from the distributed implementation viewpoint, its convergence proof  is much more involved even for i.i.d. randomization. For instance, almost sure convergence in \cite{zhao2013convergence}  depends on a stochastic approximation algorithm with expanding truncation, while our algorithm converges in the sense of $L^p$ for any $p>0$.

From the implementation point of view, the computation of (\ref{pagerankdis}) requires knowledge of the network size $N$, and the initialization $y(0)$ should be a probability vector, which is not easy to obtain with unknown network size. Finally, the PageRank problem in temporal networks with time-varying links can be addressed using the techniques provided in this paper, as discussed in the next section.

\section{The PageRank Problem for Temporal Networks with Time-Varying Links}
\label{sec_tvg}
In this section, we generalize the randomized incremental algorithms to the PageRank problem of time-varying graphs. In comparison with static graphs, we cannot simply replace $H_{s(k)}$ by $\overline{H}_{s(k)}$ in (\ref{diffusion}) due to the causality constraints, which results in the unavailability of $\overline{H}$ at time $k$. By encoding each node with a processor to record its hyperlinks, a natural way to attack this problem is to use its estimate. Thus, we obtain the following revised incremental algorithm
\bee
x(k+1)=x(k)+\alpha(k)\overline{H}_{s(k)}(k)^T(y_{s(k)}-\overline{H}_{s(k)}(k)x(k)),\label{diffusion2}
\ene
where $\overline{H}(k)$ is the estimate of $\overline{H}$ at time $k$, and is recursively computed by
\beq
\overline{H}(k)&=&\frac{\varrho^k}{1+\cdots+\varrho^k}\sum_{t=1}^k \varrho^{-t}  \cdot H(t)\nonumber\\
&=&\overline{H}(k-1)+\frac{1}{1+\cdots+\varrho^k}[W(k)-\overline{H}(k-1)].\nonumber
\enq
By (\ref{perhylkm}), it is clear that $\lim_{k\rainfty}\overline{H}(k)=\overline{H}$. Note that the algorithm in (\ref{diffusion2}) can be distributedly implemented as in (\ref{condensed}) with the same approach previously explained.

For temporal networks with time-varying links, another problem is how to appropriately define the transition probability matrix of the Markovian process $s(k)$. Obviously, the transition probability matrix, which characterizes the random surfer browsing behavior, needs to be adapted to temporal networks. Following
the approach used for static graphs, the Definition \ref{assumption1} is revised as follows.

\begin{defi}\label{assumption2}[Transition Probability Matrix for Time-Varying Graphs]
The incremental process $s(k)$ is a Markov chain with a transition probability matrix $P(k)=(1-\omega)P^{W(k)}+\frac{\omega}{N}\bone\bone^T$ where $\omega\in[0,1]$ and $P^{W(k)}$ is given by
\bee
P^{W(k)}_{ij}=
\left\{
\begin{array}{ll}
\min\{\frac{1}{D_i(k)+1},\frac{1}{D_j(k)+1}\},&~\text{if}~(i,j)\in\cE(k),\\
1-\sum_{(i,u)\in\cE(k)}P^{W(k)}_{iu},&~\text{if}~i=j,\\
0&\text{otherwise.}
\end{array}
\right.\nonumber
\ene
\end{defi}

Similar to Definition \ref{assumption1}, $P(k)$ exploits the behavior of a random surfer when browsing webpages at time $k$ in time-varying graphs, and the convex combination parameter $\omega\in[0,1]$ represents the two modes of randomization (centralized and distributed) previously described. However, the transition probability matrix is constant for fixed graphs, which implies that $s(k)$ is a homogenous Markov chain. This fact does not hold for temporal networks. In this case, we are dealing with a time-heterogeneous Markov process, which usually requires a much more involved analysis than the time-homogeneous one. However, Lemmas~\ref{lem_trans} and \ref{lem_ergodic} still hold and the convergence results can be proved under the following assumption.
\begin{assum}[Joint Strong Connectivity]\label{strgvarying} For the distributed randomization, i.e., $\omega=0$, we assume that there exists an integer $Q\ge 1$ such that the joint graph $(\cV, \cup_{l=k}^{k+Q-1} \cE(l))$ of the webpages is strongly connected for all $k$.
\end{assum}
The joint connectivity assumption ensures that every pair of pages are reachable from each other in finite time. We now state the main convergence result of the randomized incremental algorithm
(\ref{diffusion}) for temporal networks.

\begin{thm}[Convergence for Temporal Networks]\label{thm_convergence2}
Under Assumption \ref{strgvarying} on joint strong connectivity, the randomized incremental algorithm
(\ref{diffusion}) with a transition probability matrix given in Definition \ref{assumption2}  enjoys the following properties:
\begin{enumerate}
\renewcommand{\labelenumi}{\rm(\alph{enumi})}
\item (Almost sure convergence) $\lim_{k\rainfty}x(k)=\overline{x}^{ls}$ with probability one.
\item ($L^p$ convergence) $\lim_{k\rainfty}\twon{x(k)-\overline{x}^{ls}}_{L^p}=0$ for any $p>0$.
 \end{enumerate}
\end{thm}

\begin{proof} The proof is similar to that of Theorem~\ref{thm_convergence}, but  we need to re-elaborate  Lemma~\ref{lem_ergodic} for temporal networks.

Under Assumption \ref{strgvarying}, it follows from Lemma~1 in \cite{jadbabaie2003coordination} that
$$\widehat{P}(k)=\prod_{l=k}^{k+Q-1}P(l)$$ is ergodic for any $k\in\bN$.
Since the network size of webpages is assumed to be fixed, $P(l)$ can only take a finite number of values, i.e., $P(l)\in\cP$ where $\cP$ is a finite set containing all possible values of $P(l)$.  This implies that $\widehat{P}(k)$ takes a finite number of values as well. By Theorem~3 in \cite{jadbabaie2003coordination}, we further obtain that the product of any finite number of $\widehat{P}(k)$ is ergodic, i.e.,
 $\prod_{j=1}^m \widehat{P}(k_j)$ is ergodic as well for any finite $m$. Together with Theorem~2 in \cite{coppersmitha2008conditions} and the double stochasticity of $\widehat{P}(k)$, it follows that the distribution of $s(k)$ converges exponentially to a uniform distribution over the set $\cV$. This implies that Lemmas~\ref{lem_trans} and \ref{lem_ergodic}  will hold under Definition~\ref{assumption2} and Assumption~\ref{strgvarying}.  Then, we can easily establish the convergence for the incremental algorithm (\ref{diffusion2}) as for static graphs. Specifically, let
\beq
\zeta(k)=\alpha(k)\overline{H}_{s(k)}(k)^T(y_{s(k)}-\overline{H}_{s(k)}(k)x(k)).\nonumber
\enq This implies that
\bee
\overline{e}(k+1)=(I-\alpha(k)\cdot \overline{H}_{s(k)}^T\overline{H}_{s(k)})\overline{e}(k)+\zeta(k),
\ene
where $\overline{e}(k)=x(k)-\overline{x}^{ls}$.
Similarly, it can be shown that $\overline{e}(k)$ is uniformly bounded. Together with (\ref{perhylkm}) and Lemma~\ref{lem_ergodic}, it follows that $\lim_{k\rainfty}\zeta(k)=0$. Finally, the proof is completed by using Corollary \ref{cor_bound} in the Appendix. \end{proof}

\section{Simulations}
\label{sec_simulation}
In this section, we report simulation results on the distributed computation of betweenness centrality over a directed tree, and the degree, closeness centralities and PageRank over a randomly generated graph.
\subsection{Degree, Closeness, and PageRank Computation}
 \textcolor[rgb]{0,0,1}{We consider two types of graph models:  random graph model and scale-free network\cite{newman2010networks}, both of which are undirected with 50 nodes.  The connections between two nodes of the random graph are denoted by a dot, and the probability of a connection is selected as one half, see Fig.~\ref{fig_rangraph}. For this graph, the normalized degree, closeness and eigenvector (PageRank) centralities are shown in Fig.~\ref{fig_rangraph}.  In the PageRank computation, we test the distributed algorithm in (\ref{condensed}) on this randomly generated graph. As shown in Fig.~\ref{fig_stepsize}, the inverse of the stepsize $\alpha(k)$ of each node indeed converges to the network size. That is, the network size can be distributedly estimated by each node. Similar results can be observed in the scale-free network, see Fig.~\ref{fig_scale}.}

\begin{figure}
\centering
\includegraphics[width=4cm, height=3.85cm]{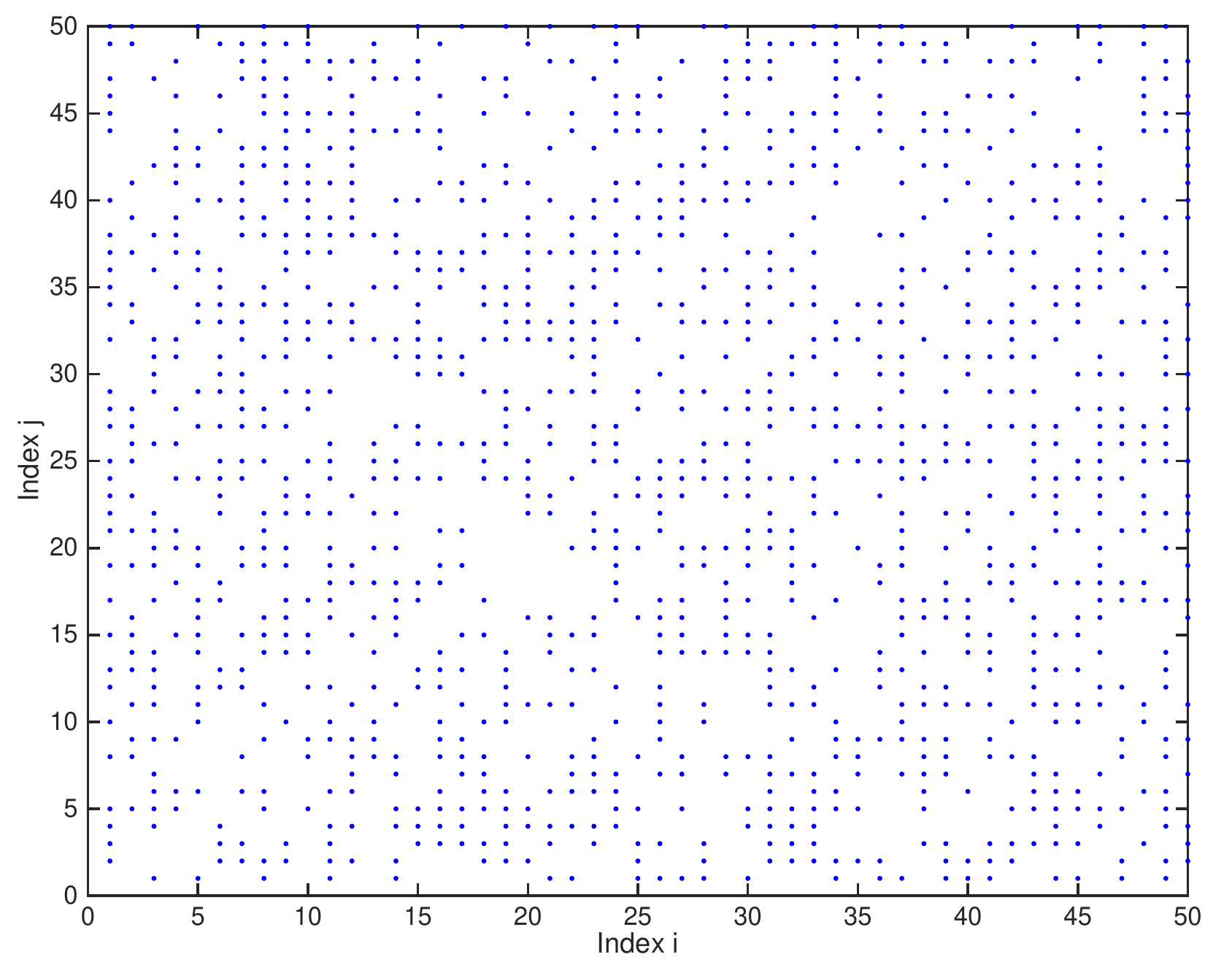}  \includegraphics[width=4.5cm, height=4cm]{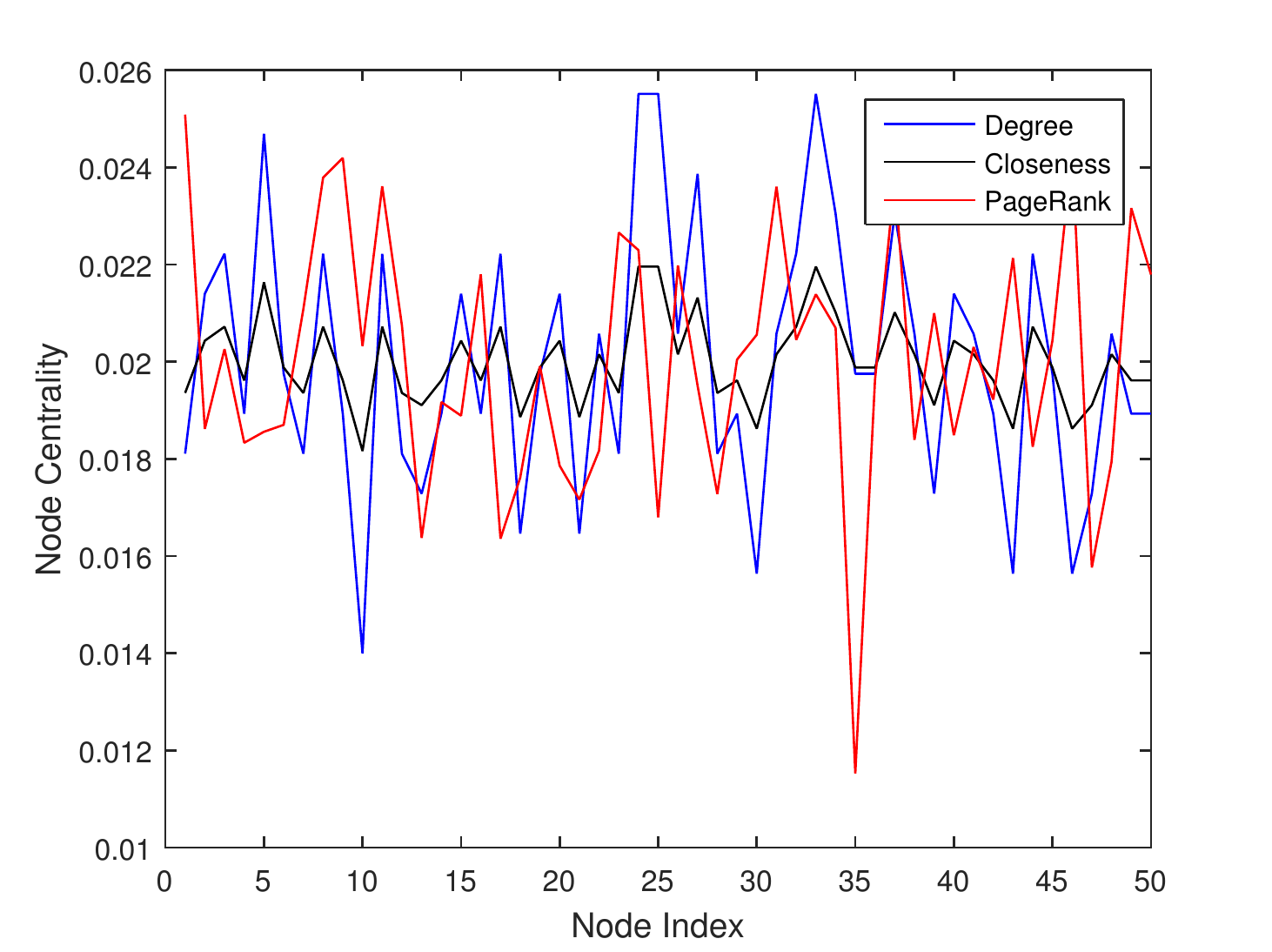}
  \caption{Random graph and its centralities.}
  \label{fig_rangraph}
\end{figure}

\begin{figure}
\centering
  \includegraphics[width=5cm, height=4cm]{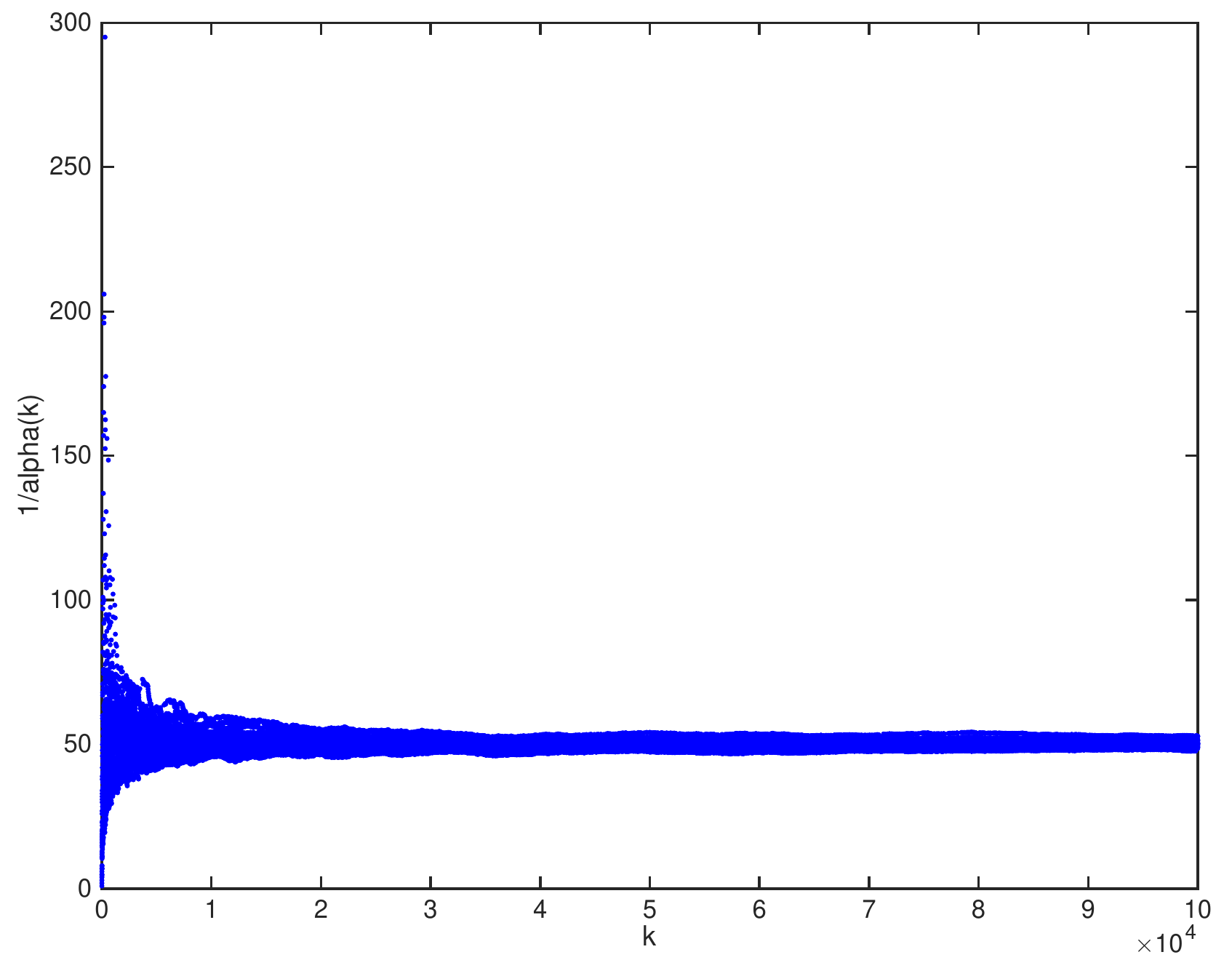}
  \caption{Stepsize.}
  \label{fig_stepsize}
\end{figure}

 \begin{figure}
\centering
\includegraphics[width=4cm, height=3.85cm]{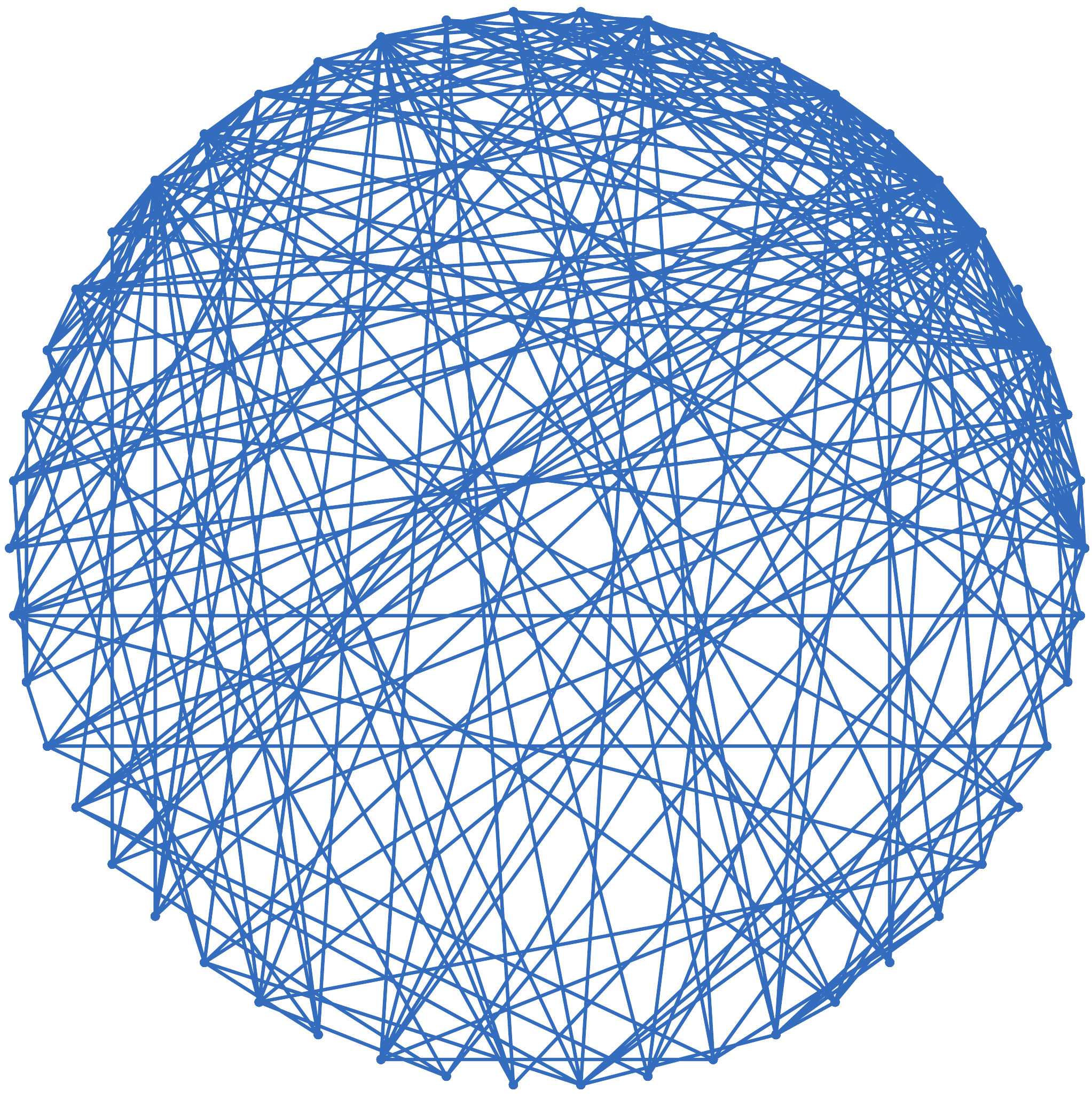}  \includegraphics[width=4.5cm, height=4cm]{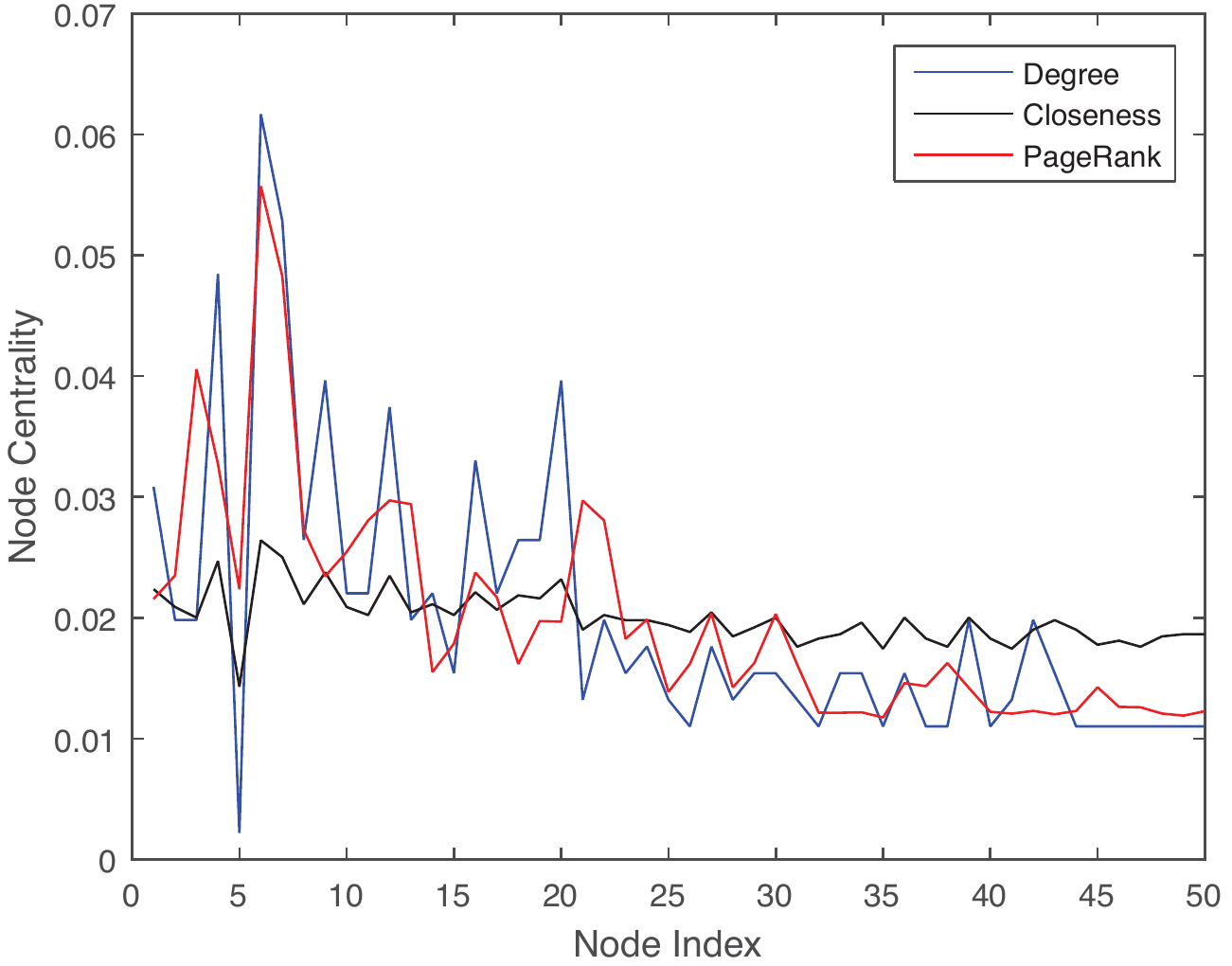}
  \caption{Scale-free network and its centralities.}
  \label{fig_scale}
\end{figure}
\subsection{Betweenness Computation Over a Directed Tree}
 We perform the distributed computation of betweenness centrality via a directed tree in Fig.~\ref{fig_tree}.
\begin{figure}
\centering
\includegraphics[width=4cm, height=4cm]{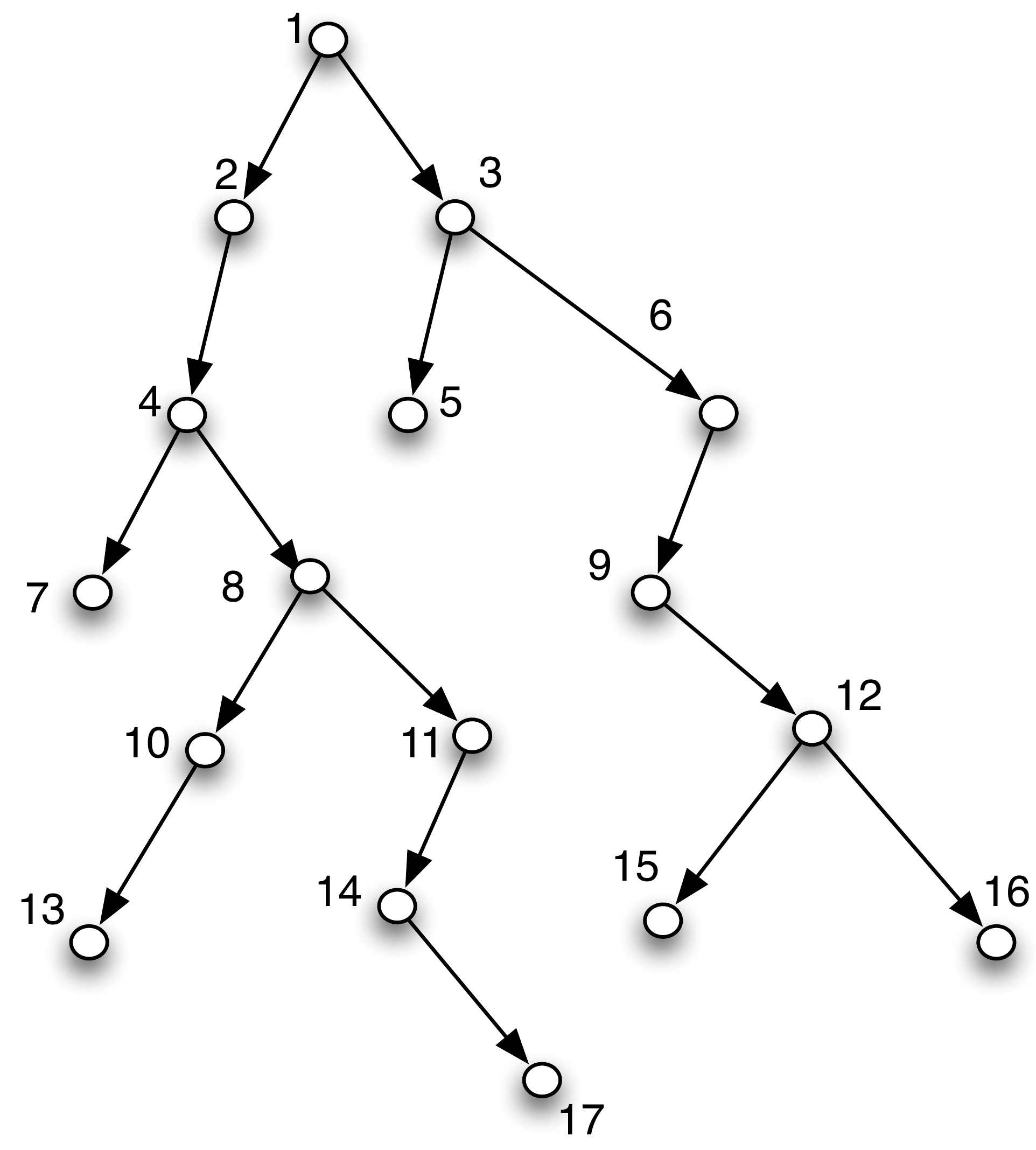}\includegraphics[width=4cm, height=4cm]{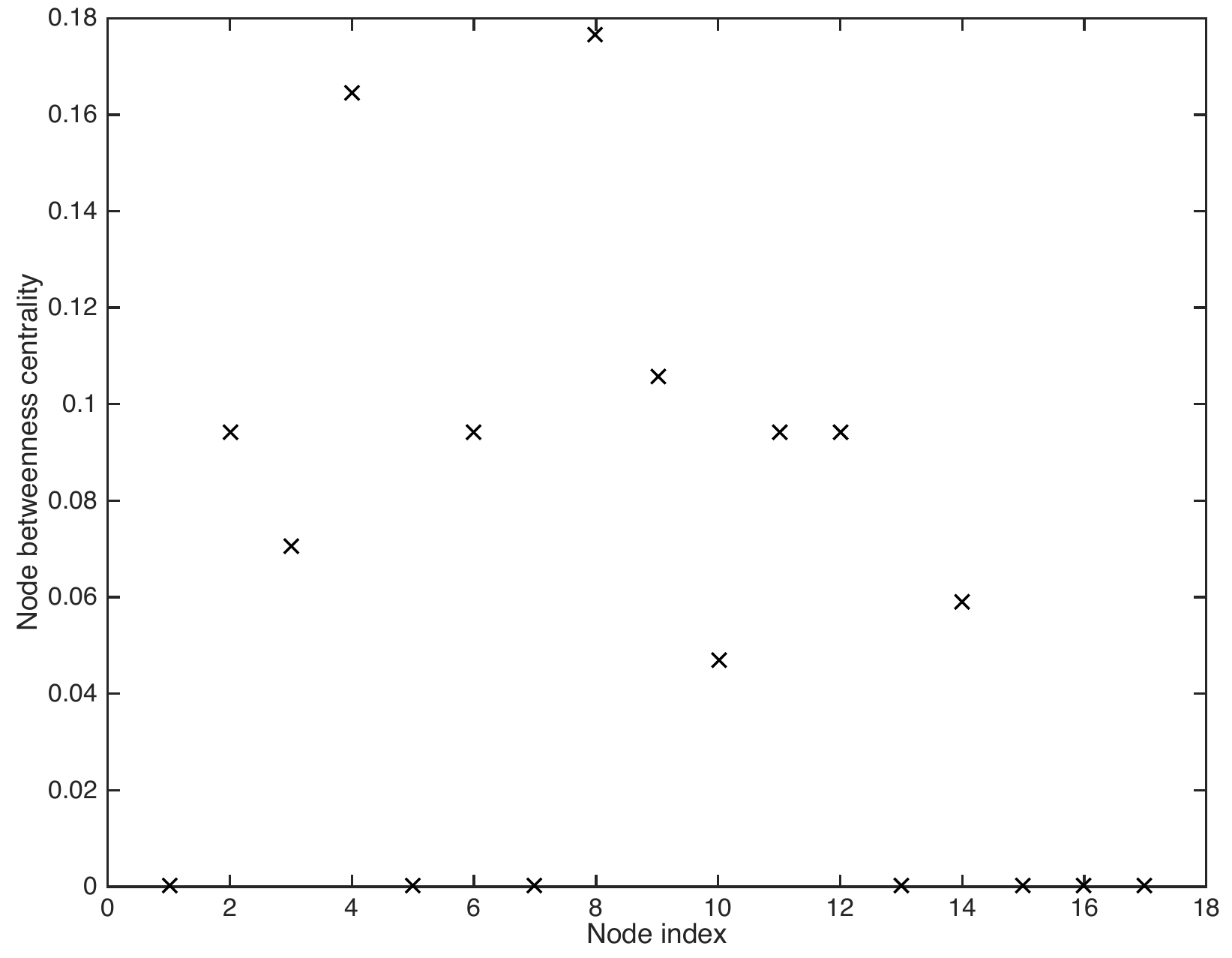}
  \caption{Directed tree and the betweenness centralities.}
  \label{fig_tree}
\end{figure}
The normalized betweenness centralities are obtained in a finite time and are illustrated in Fig.~\ref{fig_tree}, where node 8 has the largest betweenness centrality. It is consistent with our intuition that this node controls the largest number of communications between other nodes.
\subsection{PageRank Computation Over Real Web Data}
The Web data utilized in the simulations has been obtained from the database \cite{scrg} collected by crawling Web pages of various universities. This database has been previously used as a benchmark for testing PageRank algorithms \cite{ishii2014PageRank}. We have selected the data from Lincoln University in New Zealand for the year 2006. This web has 3,756 nodes and 31,718 links corresponding to 684 domains. The largest is the main domain of the university (www.lincoln.ac.nz), consisting of 2,467 pages. Other larger domains in this dataset contained, for example, 221, 101, 68, 24 pages. In this Web, a fairly large portion of the nodes are dangling nodes, e.g., there are 3,255 dangling nodes (more than $85\%$ of the total). Furthermore, two nodes with no incoming links have been removed because their effects on the PageRank values are negligible. The pages were indexed according to the domain/directory names in alphabetic order.

To proceed with the PageRank computation, the Web was modified to obtain a column stochastic matrix. This modification was done adding back links to dangling nodes. The resulting Web had 40,646 links. For this Web, the PageRank values were calculated (for comparison purposes) by the power method. About 40 iterations were sufficient for its convergence. We also implemented the PageRank Algorithm~2 in (\ref{condensed}) with unknown Web size. The results are shown in Fig.~\ref{fig_rank} and are similar to those obtained in \cite{ishii2014PageRank}. In particular, we observe that the pages with higher PageRank values are those corresponding to two clusters (with page indices around 500 and 2,500) where many pages are linked to each other. The top two pages of PageRank values are in fact the ``search" pages of the Lincoln University while the main website of the university is ranked third.

Convergence of the algorithm is guaranteed by Theorem~\ref{thm_convergence1}. The algorithm  basically has exponential convergence (Remark ~8), contrary to the algorithms proposed in \cite{ishii2014PageRank} which are based on time-averaging and have linear convergence rate. Furthermore, the algorithm does not require a common clock (Remark~3) and the Webpage randomization may be of Markovian-type instead of i.i.d., see Section V and Remark~7 in particular.

\begin{figure}
\centering
\includegraphics[width=5cm, height=4cm]{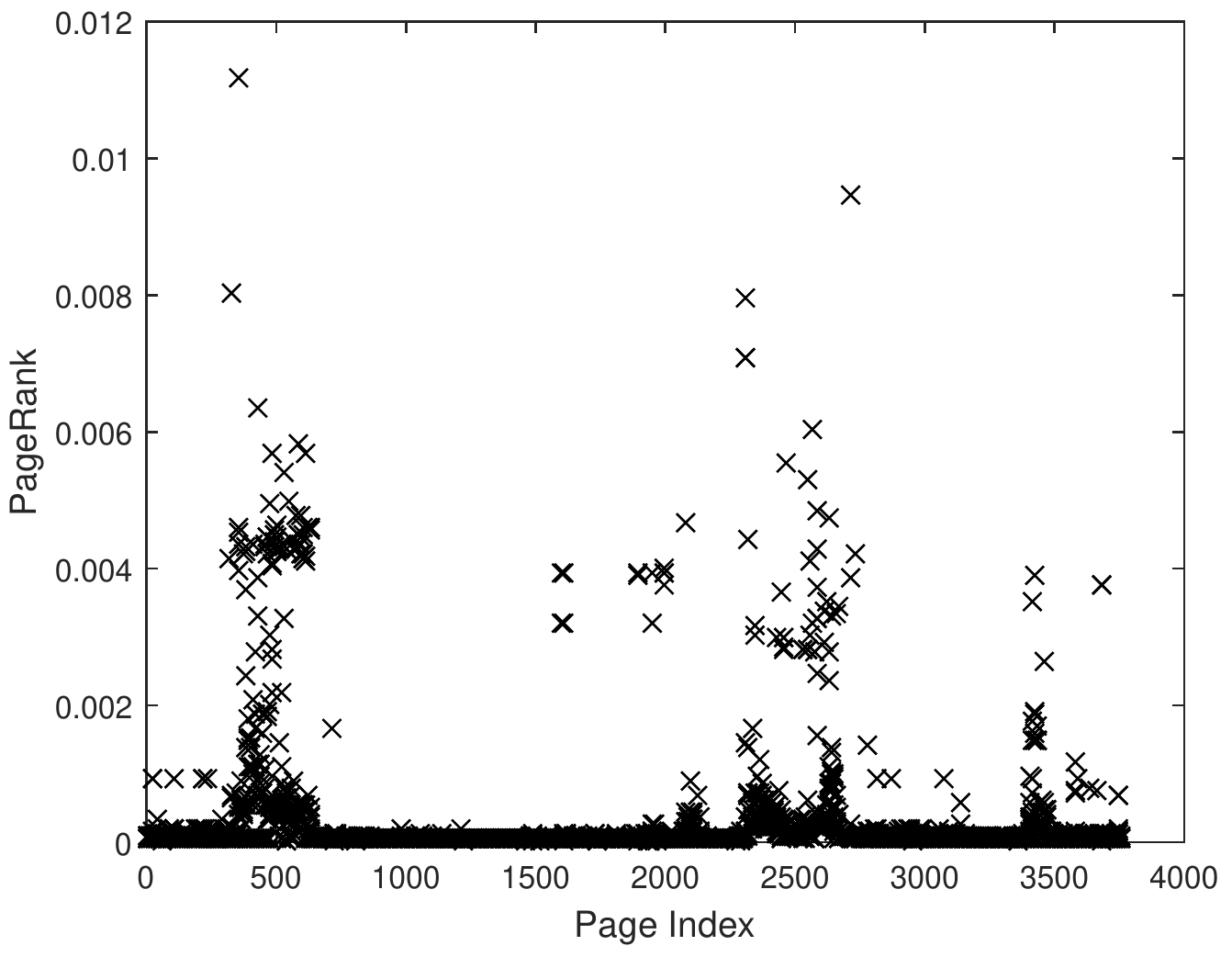}
  \caption{The PageRank of Web data.}
  \label{fig_rank}
\end{figure}

\section{Conclusion}
\label{sec_conclusion}
In this paper, we have studied the distributed computation for the degree, closeness, betweenness centrality measures and PageRank. In particular, we proposed  deterministic algorithms which converge in finite time for the degree, closeness, and betweenness centrality measures. For the PageRank problem, a randomized algorithm was devised to incrementally compute the PageRank. Different from the existing literature, this algorithm does not require to know the network size, which is typically difficult to obtain in a distributed way, and can be asymptotically estimated by each node. Extensive simulations using a classical benchmark were included to illustrate the theoretical results. Future work will be focused on extensions of this approach to study
other important applications in networks, including clock
synchronization in wireless networks \cite{freris2010fundamentals} and opinion dynamics
in social networks \cite{friedkin1999social}.

\section*{Appendix}
By using Definition 1 and Assumption 1, the random process $s(k)$ exponentially converges to a uniform stationary distribution over the set $\cV$. Without loss of generality, in this Appendix, we assume that the Markov process $s(k)$ 
already starts from the uniform distribution.

\subsection{Proof of Lemma~\ref{lem_trans}}
Before giving the proof, let $\cS_p(\lambda)$ be a set of random sequences $X=\{X_k\}$  with $\lambda\in[0,1)$ and $p\ge 1$ as
\bee
\cS_p(\lambda):=\left\{X: \twon{\prod_{j=i}^k(I-X_j)}_{L^p}\le M\lambda^{k-i}\right\}\label{exciting}
\ene
for some $M>0$, and for all $ k\ge i\ge 0$.

\begin{proof} (a) Under  Assumption \ref{assumption1}, it is easy to verify that $s(k)$ is an ergodic and irreducible process. In addition, $s(k)$ is stationary. This immediately implies that it is also a $\varphi$-mixing process \cite{billingsley1999convergence}.   By (\ref{pagerank}) and (\ref{leastsquare}), it
follows that
\beq
H_iH_i^T &\le & e_ie_i^T+(1-m)^2
W_iW_i^T\nonumber\\
&=&1+(1-m)^2\sum_{j=1}^N w_{ij}^2\nonumber\\
&< & 1+ \sum_{j\in \cL_i^1} w_{ij} \nonumber\\
&\le& N, \label{norm}
\enq
which implies that $0\le 1/N \cdot H_{s(k)}^TH_{s(k)} \le I$.

Let $\cF_k$ be the $\sigma$-algebra generated by  events associated to random variables $\{s(0),\ldots,s(k)\}$ and $\lambda_{\min}(X)$ be the minimum eigenvalue of a positive semi-definite matrix $X$.
Since $s(k)$ has a uniform distribution over $\cV$, we obtain that
\bee
1/N\cdot \bE[ H_{s(k)}^TH_{s(k)}]=1/N\cdot \cH,
\ene
where $\cH$ is defined in (\ref{meanh}).

By Lemma~\ref{lem_definite}, all the eigenvalues of $1/N\cdot\bE[H_{s(k)}^TH_{s(k)}]$ are strictly positive.  In light of Theorem~2.3 in  \cite{guo1994stability}, there exists a $h>0$ and $\lambda \in(0, 1)$ such that $\{\lambda_k\}\in \cS_1(\lambda)$ where
$$\lambda_k:=\lambda_{\min}\left\{\bE\left[\frac{1}{(h+1)N}\sum_{i=kh+1}^{(k+1)h}H_{s(i)}^TH_{s(i)}|\cF_{kh}\right]\right\}$$
 and $\cS_1(\lambda)$ is defined in (\ref{exciting}). Together with Theorem~2.1 in  \cite{guo1994stability}, it follows that
$$\{\frac{1}{N}\cdot H_{s(k)}^TH_{s(k)} \}\in \cS_2(\lambda^\alpha)~\text{and}~ \alpha=\frac{1}{8h(1+h)^2}.$$
That is, there exists a positive $M>0$ and $\rho=\lambda^{2\alpha}\in(0,1)$ satisfying (\ref{convergencerate}), which completes the proof of part (a).

(b) For ease of notation, let $t=\lfloor\frac{k-j}{N}\rfloor$. Given any sufficiently large $k$, it follows from $\twon{AB}\le \twon{A}\twon{B}$ and $\twon{\Phi(k,j+tN)}\le 1$ that
\bee
 \twon{\Phi(k,j)}\le \prod_{i=1}^t \twon{\Phi(j+i \cdot N,j+(i-1)N} \nonumber
\ene
which immediately implies that
\bee\label{division}
\frac{1}{t}\cdot \log \twon{\Phi(k,j)} \le \frac{1}{t}\sum_{i=1}^t \log\twon{\Phi(j+i \cdot N,j+(i-1)N)}.
\ene

Since $s(k)$ is ergodic and stationary, it follows from the Ergodic Theorem~\cite{ash2000pam} that
\bee\label{ergodic}
\lim_{t\rainfty} \frac{1}{t}\sum_{i=1}^t \log\twon{\Phi(j+i \cdot N,j+(i-1)N)}=\bE[\log\twon{\Phi(N,0)}]
\ene
with probability one.

It is clear from Lemma~\ref{lem_definite} that $\cH$ is positive definite. Then $\text{span}\{H_1^T,\ldots, H_N^T\}=\bR^N$, i.e., $\{H_1^T,\ldots,H_N^T\}$ generates $\bR^N$.  \textcolor[rgb]{0,0,1}{Let $\twon{\Phi_0}=\twon{\prod_{i=1}^N (I-1/N \cdot H_i^TH_i)}$. Jointly with Lemma ~3.52 in \cite{costa2005dtm}, we obtain that
$\twon{\Phi_0}<1.$}
By the ergodicity and stationarity of $s(k)$, it follows that
$$
p_0:=\bP\{s(1)=1,\ldots,s(N)=N\}>0.
$$
Note that $\twon{\Phi(N,0)}\le 1$ under any realizations of $s(1),\ldots,s(N)$, we obtain that
\bee
\bE[\log\twon{\Phi(N,0)}] \le p_0 \log \twon{\Phi_0}<0.
\ene

  For any $\epsilon \in(0,  -p_0 \log \twon{\Phi_0})$,  it follows from (\ref{ergodic}) that there exists a sufficiently large $t_0>0$ and $t>t_0$ such that
\bee
\frac{1}{t}\sum_{i=1}^t \log\twon{\Phi(j+i \cdot N,j+(i-1)N)}\le p_0 \log \twon{\Phi_0}+\epsilon.
\ene
Let $\eta =\exp( p_0 \log \twon{\Phi_0}+\epsilon)<1$ and $k_0=t_0N$, it follows from (\ref{division}) that $
\twon{\Phi(k,j)}\le \eta^t $ if $k-j>k_0$.

(c)~For any sufficiently large $k$, let $t$ be the smallest integer such that $t\ge {k}/{k_1}-1$, i.e. $t=\lceil {k}/{k_1}-1\rceil$, where $k_1=k_0+1$. This implies that $k-tk_1< k_1$. Noting that $\twon{\Phi(k,j)}\le 1$, it follows from part (b) that
\beq
\sum_{j=0}^k\twon{\Phi(k,j)}&\le& \sum_{j=0}^{k-tk_1}\twon{\Phi(k,j)}+k_1\sum_{j=1}^t\twon{\Phi(k,k-jk_1)}\nonumber\\
&\le &k_1+k_1\sum_{j=1}^t \eta^{\lfloor\frac{jk_1}{N}\rfloor}\nonumber\\
&\le &k_1\left(1+\frac{1}{1-\eta}\right).  \nonumber
\enq

 \textcolor[rgb]{0,0,1}{That is, $\sum_{j=0}^k\twon{\Phi(k,j)}$ is uniformly upper bounded by $k_1(1+\frac{1}{1-\eta})$},
which completes the proof.
 \end{proof}
 \subsection{Proof of Theorem~\ref{thm_convergence1}}
 In light of Lemmas~\ref{lem_trans}-\ref{lem_ergodic}, we have the following result, which is central to the proof of Theorem~\ref{thm_convergence1}.

\begin{cor}\label{cor_bound}
Let the transition matrix be $\widetilde{\Phi}(k+1,j)=(I-\alpha(k) \cdot H_{s(k)}^TH_{s(k)})\widetilde{\Phi}(k,j), j\le k$ and $\widetilde{\Phi}(k,k)=I$.  With probability one, it holds that
\bee\label{bound_trans1}
\lim_{k\rainfty}\sum_{j=0}^k\twon{\widetilde{\Phi}(k,j)}<\infty.
\ene
\end{cor}
\begin{proof}
We first note from (\ref{norm}) that there exists a sufficiently small $\epsilon>0$ such that $H_iH_i^T+\epsilon \le N$. In particular,
$$0<\epsilon < N-1-(1-m)^2
\min_{i\in\cV}\{\sum_{j=1}^Nw_{ij}^2\}.$$
By Lemma \ref{lem_ergodic}, there exists a sufficiently large $k_2$ such that
\bee
\frac{1}{N+\epsilon} \le \alpha(k)\le \frac{1}{N-\epsilon}, \forall k>k_2.
\ene
Combining the above relations, we obtain that
$$
0\le \alpha(k)H_{s(k)}^TH_{s(k)}\le I, \forall k>k_2.
$$
As in part (b) of Lemma~\ref{lem_trans},  there exist sufficiently large $k_3>k_2, j_1>k_2$ and $\eta_1\in(0, 1)$  such that
 \bee
\twon{\widetilde{\Phi}(k,j)}\le \eta_1^{\lfloor\frac{k-j}{N}\rfloor}~\text{if}~ k-j>k_0~\text{and}~j>j_1. \label{trans1}
\ene

Similar to part (c) of Lemma~\ref{lem_trans},  \textcolor[rgb]{0,0,1}{there exists a positive $M_{j_1}$ such that }
\bee\label{bound_trans1}
\sum_{j=j_1+1}^k\twon{\widetilde{\Phi}(k,j)}<\textcolor[rgb]{0,0,1}{M_{j_1}}, \forall k>j_1.
\ene

\textcolor[rgb]{0,0,1}{For all $ k> j_1$}, we obtain that
\beq
\sum_{j=0}^k\twon{\widetilde{\Phi}(k,j)}&=&\sum_{j=0}^{j_1}\twon{\widetilde{\Phi}(k,j)}+\sum_{j=j_1+1}^k\twon{\widetilde{\Phi}(k,j)}\nonumber\\
&\le&\sum_{j=0}^{j_1}\twon{\widetilde{\Phi}(j_1,j)}+\sum_{j=j_1+1}^k\twon{\widetilde{\Phi}(k,j)}\nonumber\\
&<&\textcolor[rgb]{0,0,1}{j_1+1+M_{j_1}},
\enq
where we used the fact that $\twon{\widetilde{\Phi}(k,j)}\le 1$ for all $k\ge j$.
\end{proof}
{\em Proof of Theorem~\ref{thm_convergence1}}: (a)~By (\ref{pagerankvector}) and (\ref{leastsquare}), it is obvious that $y_{s(k)}=H_{s(k)}x$.  Let $e(k):=x(k)-x$ and $\widetilde{y}_{s(k)}:=\widehat{y}_{s(k)}-m/N$, it follows from (\ref{diffusion1}) that
\bee
e(k+1)=(I-\alpha(k) \cdot H_{s(k)}^TH_{s(k)})e(k)+\alpha(k)H_{s(k)}^T\widetilde{y}_{s(k)},\label{error1}
\ene
which can also be written as
\bee\label{error}
e(k+1)=\widetilde{\Phi}(k+1,0)e(0)+\sum_{j=0}^{k}\widetilde{\Phi}(k,j)\alpha(j)H_{s(j)}^T \widetilde{y}_{s(j)}.
\ene

By (\ref{trans1}), it is obvious that the first term of the right hand side converges to zero with probability one.

By Lemma~\ref{lem_ergodic}, it follows that $\lim_{k\rainfty}\alpha(k)=1/N$ and $\lim_{k\rainfty}\widetilde{y}_{s(k)}=0$ with probability one. Jointly with Corollary \ref{cor_bound}, it follows from Toeplitz's Lemma~\cite{ash2000pam} that the second term in the right hand side of (\ref{error}) converges to zero with probability one.

(b)~By the first part, we know that $\sup_{k\in\bN}\twon{e(k)}<\infty$. Together with the Dominated Convergence Theorem~\cite{ash2000pam}, it follows that
\bee
\lim_{k\rainfty}\twon{x(k)-x}_{L^p}=\left(\bE[\lim_{k\rainfty}\twon{x(k)-x}_p^p]\right)^{1/p}=0,
\ene
which completes the proof. \qed

\section*{Acknowledgement}
The authors would like to thank Ms. Dan Wang from HKUST for her help on the simulations, and thank the Associate Editor and anonymous reviewers for their very constructive comments, which contributed greatly to the improvement of the quality of this work.
\bibliographystyle{IEEEtrans}
\bibliography{mybibf}
\begin{IEEEbiography}
{Keyou You}  received the B.S. degree in Statistical Science from Sun Yat-sen University, Guangzhou, China, in 2007 and the Ph.D. degree in Electrical and Electronic Engineering from Nanyang Technological University (NTU), Singapore, in 2012.

From June 2011 to June 2012, he was with the Sensor Network Laboratory at NTU as a Research Fellow. Since July 2012, he has been with the Department of Automation, Tsinghua University, China as an Assistant Professor. He held visiting positions at Politecnico di Torino, Hong Kong University of Science and Technology, University of Melbourne and etc. His current research interests include networked control systems, parallel networked algorithms, and their applications.

He received the Guan Zhaozhi best paper award at the 29th Chinese Control Conference in 2010, and a CSC-IBM China Faculty Award in 2014. He was selected to the national ``1000-Youth Talent Program" of China in 2014.\end{IEEEbiography}

\begin{IEEEbiography}
{Roberto Tempo}   is currently a Director of Research of Systems and Computer Engineering at CNR-IEIIT, Politecnico di Torino, Italy. He has held visiting positions at Chinese Academy of Sciences in Beijing, Kyoto University, The University of Tokyo, University of Illinois at Urbana-Champaign, German Aerospace Research Organization in Oberpfaffenhofen and Columbia University in New York. His research activities are focused on the analysis and design of complex systems with uncertainty, and various applications within information technology. On these topics, he has published more than 180 research papers in international journals, books and conferences. He is also a co-author of the book ``Randomized Algorithms for Analysis and Control of Uncertain Systems, with Applications'', Springer-Verlag, London, published in two editions in 2005 and 2013.

Dr. Tempo is a Fellow of the IEEE and a Fellow of the IFAC. He is a recipient of the IEEE Control Systems Magazine Outstanding Paper Award, of the Automatica Outstanding Paper Prize Award, and of the Distinguished Member Award from the IEEE Control Systems Society. He is a Corresponding Member of the Academy of Sciences, Institute of Bologna, Italy, Class Engineering Sciences.

In 2010 Dr. Tempo was President of the IEEE Control Systems Society. He is currently serving as Editor-in-Chief of Automatica. He has been Editor for Technical Notes and Correspondence of the IEEE Transactions on Automatic Control in 2005-2009 and a Senior Editor of the same journal in 2011-2014. He is a member of the Advisory Board of Systems \& Control: Foundations \& Applications, Birkhauser. He was General Co-Chair for the IEEE Conference on Decision and Control, Florence, Italy, 2013 and Program Chair of the first joint IEEE Conference on Decision and Control and  European Control Conference, Seville, Spain, 2005.\end{IEEEbiography}
\begin{IEEEbiography}
{Li Qiu} received his Ph.D.\ degree in electrical engineering from the
University of Toronto in 1990. After briefly working in the Canadian Space Agency, the Fields Institute for Research in Mathematical Sciences (Waterloo),
and the Institute of Mathematics and its Applications (Minneapolis), he joined Hong Kong
University of Science and Technology in 1993, where he is now a
Professor of Electronic and Computer Engineering.

Prof.\ Qiu's research interests include system, control,
information theory, and mathematics for information technology, as well as their
applications in manufacturing industry and energy systems.
He is also interested in control education and coauthored an undergraduate textbook
``Introduction to Feedback Control'' which was published by Prentice-Hall in 2009.
This book has so far had its North American edition, International edition, and Indian edition. The Chinese Mainland edition is to appear soon. He served as an
associate editor of the {\em IEEE Transactions on Automatic Control} and an associate editor of {\em Automatica}. He was the general chair of the 7th Asian Control Conference, which was held in Hong Kong in 2009. He was a Distinguished Lecturer from 2007 to 2010 and is a member of the Board of Governors in 2012 of the IEEE Control Systems Society. He is the founding chairperson of the Hong Kong Automatic Control Association, serving the term 2014-2017. He is a Fellow of IEEE and a Fellow of IFAC.

\end{IEEEbiography}

\end{document}